\documentclass{IEEEtran}
\usepackage{graphicx}
\usepackage{comment}
\usepackage{enumitem}
\usepackage{amsmath,amsthm, amsfonts,amssymb,epsfig,epstopdf}
\usepackage[linesnumbered, vlined, ruled]{algorithm2e}
\usepackage{url}
\usepackage{soul}
\usepackage[usenames,dvipsnames]{color}
\newtheorem{observation}{Observation}[section]
\newtheorem{lemma}{Lemma}[section]
\newtheorem{corollary}{Corollary}[section]
\newtheorem{fact}{Fact}[section]
\newtheorem{prop}{Proposition}[section]

\newtheorem{definition}{Definition}[section]
\newtheorem{theorem}{Theorem}[section]
\graphicspath{{./results/}{./images/}}

\usepackage{subfigure,epsfig,epstopdf}
\usepackage{xcolor}
\usepackage{xfrac}
\usepackage{wrapfig,graphicx}
\usepackage{enumitem}
\usepackage{algorithmic}
\usepackage{booktabs}

\usepackage{etoolbox}
\let\bbordermatrix\bordermatrix
\patchcmd{\bbordermatrix}{8.75}{4.75}{}{}
\patchcmd{\bbordermatrix}{\left(}{\left[}{}{}
\patchcmd{\bbordermatrix}{\right)}{\right]}{}{}

\def\psiAPR{\psi^{\mbox{\tiny AR}}}
\def\psiCPR{\psi^{\mbox{\tiny CR}}}
\def\psiTREE{\psi^{\mbox{\tiny SS}}}
\def\psiFOREST{\psi^{\mbox{\tiny MS}}}

\usepackage{tikz}
\newcommand{\LD}{\langle}
\newcommand{\RD}{\rangle}

\begin{document}
\title{On Fundamental Bounds of Failure Identifiability by Boolean Network Tomography}
\author{ 
Novella Bartolini, {\em Senior Member IEEE},
Ting He, {\em Senior Member IEEE}, \\
Viviana Arrigoni,
Annalisa Massini, and
Hana Khamfroush
       \IEEEcompsocitemizethanks{
\IEEEcompsocthanksitem 
N. Bartolini, V. Arrigoni and A. Massini are with the Department
of Computer Science, Sapienza University of Rome, Italy. 
E-mail: \{bartolini,arrigoni,massini\}@di.uniroma1.it
\IEEEcompsocthanksitem T. He is with the Department of Computer Science and Engineering, Pennsylvania State University, USA.  
E-mail: tzh58@cse.psu.edu.
 \IEEEcompsocthanksitem H. Khamfroush is with the Department of Computer Science, University of Kentucky, USA.
 E-mail: khamfroush@cs.uky.edu 
%
\IEEEcompsocthanksitem This work was supported by the Defense Threat Reduction Agency under the
grant HDTRA1-10-1-0085, and by NATO
 under
the SPS grant G4936 SONiCS.
}}

\maketitle

\begin{abstract}
Boolean network tomography is a powerful tool to infer
the state (working/failed) of individual nodes from path-level measurements obtained by egde-nodes. 
We consider the problem of optimizing the capability of identifying network failures 
through the design of monitoring schemes. Finding an optimal solution is NP-hard  and a large body of work has been devoted to heuristic approaches providing lower bounds. 
Unlike previous works, we provide upper bounds on the maximum number of identifiable nodes, given the number of monitoring paths and different constraints on the network topology, the routing scheme, and the maximum path length.
These upper bounds
represent a fundamental limit on  identifiability of failures via Boolean network tomography.
Our analysis provides insights 
on how to design topologies and related monitoring schemes to achieve the maximum identifiability under various network settings. Through analysis and experiments we demonstrate the tightness of the bounds and efficacy of the design insights for engineered as well as real networks.
 \end{abstract}

\section{Introduction and motivation}
\label{sec:intro}
The capability to assess the states of network nodes in the presence of failures is fundamental for many functions in network management, including performance analysis, route selection, and network recovery. In modern networks, the traditional approach of relying on built-in mechanisms to detect node failures is no longer sufficient, as bugs and configuration errors in various customer software and network functions often induce ``silent failures'' that are only detectable from end-to-end connection states
\cite{Kompella07infocom}.
\emph{Boolean network tomography} \cite{Duffield03} is a powerful tool to infer the states of individual nodes of a network from binary measurements taken along selected paths. 
We consider the problem of Boolean network tomography in the framework of graph-constrained group testing~\cite{Cheraghchi:TOI2012}.
Classic group testing \cite{Dorfman:annals,Atia:TOIT2012}   studies  the  problem of identifying defective items in a large set $S$ by means of binary measurements taken on subsets $S_i \subseteq S$ ($i=1, \ldots, m$).
Close to the problem of group testing, Boolean network tomography  aims at identifying defective network items, i.e. nodes or links,  in a large set $S$ including all the network components,
by performing binary measurements over subsets $S_i$, i.e., monitoring paths. As in graph-based group testing, the composition of the testing sets  conforms to the structure of the network. 
In this regard, Cheraghchi et al. \cite{Cheraghchi:TOI2012}   studied graph-constrained group testing with the goal of minimizing the number of monitoring paths needed to identify the state (defective or normal) of all network nodes, under the assumption that the maximum number of defective nodes is given.
In their work, paths are defined by random walks in the graph, and the authors give upper bounds on the number of paths needed.


In our work, we tackle the problem of maximizing the number of nodes whose states can be uniquely determined from binary measurements on a given number of monitoring paths.
Unlike \cite{Cheraghchi:TOI2012}, we consider that 
monitoring paths are constrained not only by the network topology, but also by the routing scheme adopted in the network, and by additional requirements in case of passive monitoring, i.e. monitoring paths coinciding with some service related paths.

Due to the inherent hardness in computing the exact maximum value, we focus on deriving easily computable upper bounds which allow us to: (i) evaluate the room of improvement for a given monitoring scheme in a specific network setting, and (ii) extract rules for network design to maximize the number of identifiable nodes in a general setting.

The main contributions of this work are the following:
\begin{itemize}
\item
We upper-bound the maximum number of identifiable nodes with  a given number of monitoring paths, in the following scenarios: (1)
paths between arbitrary nodes under arbitrary routing (Theorem~\ref{th:exp_max});
(2)
paths between arbitrary nodes under consistent routing (Theorem \ref{th:bound_consistent_routing});
{\color{black}{(3)
paths between arbitrary nodes under partially consistent routing (Theorem \ref{th:bound_hcr});}}
(4) paths from a single server to multiple clients under consistent routing (Theorem \ref{th:bound_with_tree});
(5) paths from multiple servers to multiple clients with fixed/flexible assignment under consistent routing (Theorems \ref{th:bound_with_services} and  \ref{th:undistinguished_clients_rev}).
\item We give insights on the design of topologies and  monitoring schemes to approximate the bounds, grounded upon the bound analysis.
\item We demonstrate the tightness of the  upper bounds by providing constructive approaches and comparisons with the results of known heuristics \cite{usICDCS16} on engineered as well as real network topologies.
\item We compare the bounds in different scenarios to evaluate the impact of the routing scheme, the number of monitoring paths, and the maximum path length on the number of identifiable nodes.
\end{itemize}

\section{Related work}
\vspace{-.cm}
Pioneered by Duffield  \cite{Duffield03}, Boolean network tomography has direct applications in network failure localization. The early works focused on best-effort inference. For example, Duffield et al. \cite{Duffield03,Duffield06TInfo} and Kompella et al. \cite{Kompella07infocom} aimed at finding the minimum set of failures that can explain the observed measurements, and Nguyen et al. \cite{Nguyen07infocom} aimed at finding the most likely failure set that explains the observations.
Later, the identifiability problem attracted attention. Ma et al. characterized in \cite{Ma&etal:14IMC} the maximum number of simultaneous failures that can be uniquely localized, and then extended the results in \cite{Ma16TON} to characterize the maximum number of failures under which the states of specified nodes can be uniquely identified as well as the number of nodes whose states can be identified under a given number of failures.
{\color{black}{Galesi et al. \cite{Galesi}
study upper and lower bounds on the maximum identifiability index of a topology, i.e. the maximum number of simultaneous failures under which the monitoring system is still capable of identifying the state of all network nodes.
These studies are orthogonal to ours, as we aim at bounding the number of identifiable nodes, within a given identifiability index.
The related optimization problems have also been studied.
The problem of optimally placing monitors to detect failed nodes via round-trip probing was introduced and proven to be NP-hard  by Bejerano et al. in \cite{Bejerano03INFOCOM}.
The work by Cheraghchi et al. \cite{Cheraghchi:TOI2012} aimed at determining the minimum number of monitoring paths to uniquely localize a given number of failures, under the assumption that any path can be monitored. For monitoring paths that start/end at monitors, Ma et al. \cite{Ma15PE} proposed polynomial time heuristics to deploy a minimum number of monitors to uniquely localize a given number of failures under various routing constraints.
When monitoring is performed at the service layer, He et al. \cite{usICDCS16} proposed service placement algorithms to maximize the number of identifiable nodes by monitoring the paths connecting clients and servers.

Boolean network tomography is not to be confused with robust network tomography, which aims at inferring fine-grained performance metrics (e.g., delays) of non-failed links under failures. For robust network tomography, Tati et al. \cite{Tati14ICDCS} proposed a path selection algorithm to maximize the expected rank of successful measurements subject to random link failures, and Ren et al. \cite{Ren:Infocom16}  proposed algorithms to determine which link metrics can be identified and where to place monitors to maximize the number of identifiable links, subject to a bounded number of link failures.
Robust network tomography has also been studied under settings not limited to failures  \cite{ting1,ting2} to study the identifiability of additive link metrics under topology changes. 
%
%


Our work  addresses the problem of maximizing the number of identifiable nodes under failures. It extends a 
 previous work \cite{BartINFOCOM2017} with improved bounds, new design techniques and characterization of monitoring topologies. 
 
\section{Problem formulation}\label{sec:Problem Formulation}
Throughout the paper we use the definitions given in Table~\ref{tab:notation}, and we use the short forms {\em wrt} for "with respect to" and {\em iff} for "if and only if".
We model the  network as an undirected graph $\mathcal{G} = (V,\: E)$, where $V$ is a set of  nodes,
and $E$ is the set of 
links.
According to the needs of the discussion, a path $p$ defined on $G$ is represented as either a \emph{set} of nodes $p$, or as an ordered \emph{sequence}
 of nodes $\hat{p}$, from one endpoint to the other.
 Each node may be in working or failed state.
  The state of a path is working if and only if all traversed nodes (including endpoints) are in working state.
Without loss of generality, we assume that links do not fail and model  network links through logical nodes so that a link failure corresponds to the failure of a logical node.
The set of \emph{all} failed nodes,  denoted by $F \subseteq V$, defines the state of a network, and is  called \emph{failure set}.

{\footnotesize{
\begin{table} [h]
\begin{center}
\begin{tabular}{| c | l |}
\hline
 \bfseries Notation & \bfseries Description \\ 
\hline 
$T$ & Testing matrix, $T \in \{ 0,1\}^{m \times n}$, for $m$ paths and $n$ nodes \\  
$P$ & Set of $m$ monitoring paths $P=\{ p_1, \ldots, p_m\}$\\
$p$, $\hat{p} \in P$ & monitoring path as a set or a list of nodes, respectively \\
 $b(v)$ & Boolean encoding of node $v$  wrt $P$\\
 $b(v)|_i$ & $i$-th element of $b(v)$ (equal to 1 iff  $v \in p_i$, to 0 otherwise)\\ 
 $\chi(v)$ & Crossing number of node $v$  wrt a set of paths $P$\\  
 $P_{F}$ & Incident set of paths  of a failure set $F$\\
 $\mathcal{I}(p)$ & Set of identifiable nodes traversed by path $p$\\
 $M(\hat{p})$ & Path matrix of path $\hat{p}$ \\
 $\mathcal{B}$ & Set of all the node encodings in $\{0,1 \}^m$, with $m$ paths\\
$\mathcal{B}|_i \subset \mathcal{B}$ & $\{ b \in \mathcal{B},\ s.t.\ b|_i=1 \}$, $i=1, \ldots, m$\\
$\mathcal{B}(k) \subset \mathcal{B}$ & $\{ b \in \mathcal{B},\ s.t.\ \sum_{i=1}^m b|_i=k \}$, $k=1, \ldots, m$\\
$\ell_i(B)$ & $\ell_i(B)= |B\cap\mathcal{B}|_i|$, where $B \subseteq \mathcal{B}$\\
 \hline
\end{tabular}
\end{center}
 \caption{Notation table} \label{tab:notation}
\end{table}
}}
We assume that node states cannot be measured directly, but only indirectly via \emph{monitoring paths}.
Let $P=\{p_1, p_2, \ldots ,p_m \}$ be a given set of $m$ monitoring paths. 
We call the {\em incident set} of $v_i$ the set of paths affected by the failure of node $v_i$ and denote it with $P_{v_i}$. 
{\color{black}We  define with $\chi(v_i) \triangleq |P_{v_i}|$, the {\em crossing number} of node $v_i$, which is the number of monitoring paths traversing $v_i$, i.e., the cardinality of its incident set.}
We also 
denote the incident set of paths of a failure set $F$ with \begin{math}P_F\triangleq \cup_{v_i \in F} P_{v_i}\end{math}.

The {\em testing matrix} $T$ is an $m \times n$ matrix, whose element $T|_{i,j}=1$ if 
node $v_j$ is traversed by path $p_i$, i.e., 
$v_j \in p_i$, and zero otherwise.
The $j$-th column of the test matrix $ T|_{*,j}$ is the characteristic vector\footnote{A {\em characteristic vector} of a subset $S$ of an ordered set of $n$ elements $V=\{v_1, v_2, \ldots, v_n\}$ is a binary vector with `1' only in the positions of the elements of $V$ that are included in $S$.} of $P_{v_j}$, hereby denoted with $b(v_j) \triangleq T|_{*,j}$ and called the  {\em binary encoding} of~$v_j$. Note that multiple nodes may have the same binary encoding.
{\color{black}
\begin{observation}
\label{obs:norm_of_b}
Consider a node $v$, and a set $P=\{p_1, \ldots, p_m \}$ of monitoring paths.
It holds that $v \in p_i$ iff the $i$-th element of its binary encoding is equal to 1, i.e., $b(v)|_i=1$;  consequently, the crossing number $\chi(v)$  is equal to the number of ones in the binary  encoding of $v$, namely $\chi(v)=
\sum_{i=1}^m b(v)|_i$.
\end{observation}
}

\subsection{Identifiability}
The concept of identifiability refers to the capability of inferring the states of individual nodes from the states of the monitoring paths. Informally, we say that a node $v$ is 1-identifiable with respect to a set of paths $P$, if its failure and the failure of any other node $w$ cause the failure of different sets of monitoring paths in $P$, i.e. $v$ and $w$ have different incident sets.
This concept can be extended to the case of  concurrent failures of at most $k$ nodes, where a node is $k$-identifiable in $P$ if any two sets of failures $F_1$ and $F_2$ of size at most $k$, which differ at least in  $v$ (i.e., one contains $v$ and the other does not), cause the failures of different monitoring paths in $P$, i.e. $F_1$ and $F_2$ have different incident sets.

He et al. in \cite{usICDCS16} formalized the concept of $k$-identifiability that we reformulate as follows:

\begin{definition}
\label{def:node_k_identifiability}
Given a set of monitoring paths $P$ and a node $v_i \in V$, $v_i$ is called {\em $k$-identifiable}  wrt $P$ when for any failure sets $F_1$ and $F_2$ such that
$F_1 \cap \{v_i \} \neq F_2 \cap \{v_i\}$, and $|F_j| \leq k$ ($j \in \{1,2\}$),
 {\color{black}the incident sets $P_{F_1}$ and $P_{F_2}$ are different. Equivalently, it holds that:}
\noindent
\begin{center}
\begin{math}
\bigvee_{v_s \in F_1} b(v_s) \neq \bigvee_{v_z\in F_2} b(v_z)
\end{math}
\end{center}
\noindent
where with "$\bigvee$" we refer to the element-wise logical ${\texttt{OR}}$.
\end{definition}


The following Lemma considers the special case of $k=1$.
\begin{lemma}\label{def:identifiability}
A node $v_i$ is 1-identifiable wrt $P$ iff 
$b(v_i)\neq \mathbf{0}$, and $\forall v_j \neq v_i$, $b(v_j)\neq b(v_i)$, i.e., its binary encoding is not null and not identical with that of any other node.
\end{lemma}

\begin{proof}
{\color{black}
Let us assume that $v_i$ is 1-identifiable, and consider Definition  \ref{def:node_k_identifiability}, for any  two sets $F_1$ and $F_2$, each with cardinality at most 1.
Without loss of generality, we consider $v_i \in F_1$, then $F_2$ is either empty or contains only one node $v_j$, such that $v_j \neq v_i$.
Therefore, Definition \ref{def:node_k_identifiability}
 implies that $b(v_i)\neq \mathbf{0}$ (if we choose $F_2 = \emptyset$) and  $b(v_j)\neq b(v_i)$, $\forall v_j \neq v_i$ (if we choose $F_2 =\{ v_j \}$).
 
Let us now assume that  node $v_i$ is such that $b(v_i)\neq \mathbf{0}$, and $\forall v_j \neq v_i$, $b(v_j)\neq b(v_i)$.
The assumption implies that $P_{v_i} \neq \emptyset$ and for any other node $v_j \neq v_i$, $P_{v_j}\neq  P_{v_i}$, i.e.,  $v_i$ is $1$-identifiable 
according to Definition \ref{def:identifiability}.
}
\end{proof}
{\color{black}We clarify that 
by Lemma \ref{def:identifiability}, a node with null encoding is not $1$-identifiable, even if its encoding were unique, which happens when it is the only non-monitored node. This is because, for a node to be considered identifiable, we must be able to assess its status, working or failed, based only on the status of the  monitoring paths, which requires the node to be traversed by at least a path.}

\subsection{Bounding identifiability}
The set of monitoring paths $P$ is usually 
the result of  design choices related to topology, monitoring endpoints, routing scheme, etc. Given a collection of candidate path sets $\mathcal{P}$ under all possible designs\footnote{For example,  $\mathcal{P}$ may be the class of path sets of given cardinality, or  paths of a given length between given sources and each of multiple candidate destinations.}, the question is: how well can we monitor the network using path measurements in $\mathcal{P}$ and which design is the best? Using the notion of $k$-identifiability, we can measure the monitoring performance by the number of nodes that are $k$-identifiable wrt $P \in \mathcal{P}$, denoted by $\phi_k(P)$, and formulate this question as an optimization: $\psi_k(\mathcal{P}) \triangleq \max_{P\in\mathcal{P}} \phi_k(P)$.

Although extensively studied 
\cite{Bejerano03INFOCOM,Cheraghchi:TOI2012,Ma15PE,usICDCS16}, the optimal solution {\color{black}is} hard to obtain due to the (exponentially) large size of $\mathcal{P}$, 
and heuristics are used to provide lower bounds. There is, however, a lack of general upper bounds. 
In this work we establish upper bounds on $\psi_k(\mathcal{P})$ in representative  scenarios.
Knowledge of these upper bounds is  key to understanding the fundamental limits of Boolean network tomography, and gives insights on the optimal network design to facilitate network monitoring.

Note that if $v_i$ is $k$-identifiable wrt $P$ for any $k\geq 1$, then $v_i$ is also $1$-identifiable wrt $P$.
\begin{lemma} \label{le:bound_on_k_or_1}
For any $k\geq 1$ and any collection $\mathcal{P}$ of 
 candidate path sets, $\psi_1(\mathcal{P})\geq \psi_k(\mathcal{P})$.
\end{lemma}
\begin{proof}
Given the optimal choice of monitoring paths $P^*\in \mathcal{P}$ achieving $\psi_k(\mathcal{P})$, we 
have $\psi_1(\mathcal{P}) \geq \phi_1(P^*) \geq \phi_k(P^*) = \psi_k(\mathcal{P})$, where the first inequality is by definition of $\psi_1(\mathcal{P})$ and the second inequality is by Definition~\ref{def:node_k_identifiability}.
\end{proof}
Therefore, in the sequel, we 
look for upper bounds on $\psi_1(\mathcal{P})$, 
 simply denoted by $\psi(\mathcal{P})$, where we will replace $\mathcal{P}$ by specific parameters in each network setting. We hereafter shortly call  the 1-identifiable nodes ``identifiable''.


\section{General network monitoring}\label{sec:novella_notes}

We initially consider a generic 
network with a given number of monitoring paths between any nodes. We analyze $\psi(\mathcal{P})$ in three cases: (i) arbitrary routing, (ii) consistent routing, {and (iii) partially-consistent routing}.
\subsection{Arbitrary routing}

\subsubsection{Identifiability bound}
Given a network with $n$ nodes, and $m$ monitoring paths,
the number of nodes that are $1$-identifiable  may grow exponentially with the number of paths.

\color{black}{
\begin{prop}
\label{prop:dmax}
Given a network with $n$ nodes, and a set of  $m$ monitoring paths $p_i,\ i=1, \ldots, m$, we denote with $\mathcal{I}(p_i)$ the set of identifiable nodes traversed by   $p_i$ and with $d_i \leq n$  the length of $p_i$ in number of nodes. It holds that  
\begin{math} |\mathcal{I}(p_i)| \leq \min \{d_i; \ 2^{m-1}\}.
\end{math}
\end{prop}}

\begin{proof}

By Lemma \ref{def:identifiability}, in order for a node to be identifiable, its binary encoding must be unique.
By Observation \ref{obs:norm_of_b}, the encodings of all the nodes traversed by path $p_i$, have a one  in the $i$-th position.
It follows that the number $|\mathcal{I}(p_i)|$ of identifiable nodes traversed by path $p_i$ is upper-bounded by its length $d_i$ and by the number of sequences  of $m$ bits (binary encodings), where the $i$-th bit is a one, which is  $2^{m-1}$.
\end{proof}
%

%
\begin{theorem}
[{\color{black}{Identifiability under arbitrary routing with known average path length}}]
\label{th:exp_max}
 Given a network with $n$ nodes, and a set $P$ of $m>1$ arbitrary routing paths, where $\bar{d}\leq n$ is the average path length, the maximum number of identifiable nodes in the network satisfies:
 \begin{align*}
 \hspace{-.5em}
 \psiAPR\hspace{-.1em}(\hspace{-.05em}m,n,\bar{d}\hspace{+0.05em})
&\hspace{-.05em}\leq \hspace{-.05em} \min \hspace{-.15em}\left\{\hspace{-.05em}\sum_{i=1}^{i_{\texttt{max}}} \hspace{-.25em} {m \choose i} \hspace{-.25em}+\hspace{-.25em}
\left\lfloor\hspace{-.25em}
\frac{N_{\texttt{max}} - \sum_{i=1}^{i_{\texttt{max}}} i \cdot {m \choose i}   }{i_{\texttt{max}}+1}
\hspace{-.25em}\right\rfloor\hspace{-.25em};n \hspace{-.10em}\right\},
\end{align*}
\noindent
where 
$i_{\texttt{max}}=\max \{
k \ | \ \sum_{i=1}^{k} i \cdot {m \choose i} \leq N_{\texttt{max}} \}$,

\noindent
and\footnote{By definition $N_{\texttt{max}}$ is an integer number.}
\noindent
$N_{\texttt{max}}= m \cdot \min \{ \bar{d}; \  2^{m-1}\}$.
\end{theorem}
\begin{proof}
{\color{black}{
The number $|\mathcal{I}(p_i)|$ of identifiable nodes traversed by a path $p_i$ of length $d_i$, $i \in \{1, \ldots, m\} $, is bounded  as described by Proposition \ref{prop:dmax}. 
Consequently,  
the number of identifiable nodes is also bounded from above as follows:  
$| \cup_{i=1}^m \mathcal{I}(p_i) | \leq \sum_{i=1}^m |\mathcal{I}(p_i)| \leq \sum_{i=1}^m \min \{d_i;\ 2^{m-1} \} \leq m \cdot \min \{ \bar{d};\ 2^{m-1}\}=N_\texttt{max}$.
%
%
%

Since we used the union bound to calculate $N_{\texttt{max}}$, this value considers some encodings multiple times when the related node belongs to  more than one path. This happens, according to  Observation \ref{obs:norm_of_b}, $\chi(v)$ times for each node $v$.

It follows that the number of distinct encodings is maximized when we minimize the number of  encoding replicas and therefore the crossing number of the related nodes.
This is achieved, within the limits of the path length, when we have ${m\choose 1}$ nodes with crossing number equal to $1$ (counted only once in $N_\texttt{max}$), ${m\choose 2}$ nodes with crossing number equal to $2$ (counted twice in $N_\texttt{max}$), and so forth, until the total number of encodings (counting the replicas) is $N_{\texttt{max}}$.

More formally, let $i_\texttt{max}=\max\{ k\ | \ \sum_{i=1}^k i \cdot {m \choose i} \leq N_\texttt{max}\}$. For each $i\leq i_\texttt{max}$, we have ${m \choose i}$ nodes with crossing number equal to $i$, i.e., traversed by $i$ paths. Considering that the remaining
$N_\texttt{max}-\sum_{i=1}^{i_\texttt{max}} i \cdot {m \choose i }$ encodings will  have at least $(i_\texttt{max}+1)$ digits equal to $1$ and thus are counted  at least $(i_\texttt{max}+1)$ times in $N_\texttt{max}$, the number of distinct encodings out of the $N_\texttt{max}$ encodings is upper-bounded by:

\begin{math}
\noindent{\small
 \psiAPR(m,n,\bar{d})
\leq  \sum_{i=1}^{i_{\texttt{max}}}  {m \choose i} +
\left\lfloor
\frac{N_{\texttt{max}} - \sum_{i=1}^{i_{\texttt{max}}} i \cdot {m \choose i}}{i_{\texttt{max}}+1}
\right\rfloor
}
\end{math}. 

\noindent
Considering also that the number of identifiable nodes cannot exceed $n$, we have  the final bound.
}
}\end{proof}

{\color{black}
We underline that Theorem \ref{th:exp_max} provides a topology-agnostic bound, i.e., a theoretical limit which is valid for any topology and only considers 
the number of nodes, the number of monitoring paths, and the average path length\footnote{As the constraints imposed by the topology of the network and path routing are not taken into account in this theorem, its validity holds also for any group testing problem where $m$ groups of known average size, are used to inspect the state of $n$ elements.} $\bar{d}$.

{\color{black}{
We observe that when paths have arbitrary unbounded length, 
we have $N_\texttt{max}=m \cdot  2^{m-1}$, 
and   $i_{\texttt{max}}= m$. In such a case, 
Theorem \ref{th:exp_max} reduces to the following corollary for unbounded path length.
\textcolor{black}{
\begin{corollary}[{\color{black}{Identifiability under arbitrary routing and unbounded path length}}]
\label{co: old_bound}
Given a network with $n$ nodes and a set $P$ of $m$ monitoring paths, the maximum number of identifiable nodes  satisfies:
\begin{center}
  \begin{math}
 \psiAPR(m,n) \leq \min \{n;2^m-1\}.
 \end{math}
 \end{center}
\end{corollary}}

}}

{\color{black}{
Notice  that it may be of interest to have a bound on the number of identifiable nodes when the average length of monitoring paths is not known but there are topology or QoS related constraints on the length of a path expressed in terms of  a maximum value $d_{\texttt{max}}$. In this case, we have the following variation of the bound due to the fact that:
$$ \bar{d} \leq \max_i \{d_i\} \leq d_{\texttt{max}}.$$

\begin{corollary}
[Identifiability under arbitrary routing and bounded maximum path length]
\label{corollary:exponential}
 Given a network and a set $P$ of $m>1$ arbitrary routing paths with maximum length $d_\texttt{max}$, the maximum number of identifiable nodes in the network is upper-bounded as in Theorem \ref{th:exp_max}, except that $N_{\texttt{max}}$ is now defined as: 
\noindent
$N_{\texttt{max}} =  m \cdot \min \{ d_\texttt{max}; \  2^{m-1}\}$.
\end{corollary}

\subsubsection{Design via Incremental Crossing Arrangement (ICA)}\label{subsubsec:ICA}

The proof of Theorem \ref{th:exp_max} suggests a  technique to build a network topology $G=(V,E)$ and related monitoring paths $P$ with maximum identifiability, where  $|P|=m$. We call this technique {\em Incremental Crossing Arrangement} (ICA).

{\em ICA, the idea. }  The  technique  works by generating node encodings in increasing order of crossing number with respect to the monitoring paths in use, until the number of generated encodings reaches the bound defined  in Theorem  \ref{th:exp_max}. Monitoring paths must be designed so as to 
traverse nodes according to  the generated encodings: path $p_i$ traverses any node $v$ for which $b(v)|_i=1$, $\forall i \in \{1, \ldots, m \}$. The  network topology is then constructed by  considering a node for each of the generated  Boolean encodings,  and adding links between any pair of nodes appearing sequentially in any  path.
 
 {\em ICA in details. }
In the following we consider an  arbitrarily  large number of nodes $n$, such that $n$ is larger than the  bound on identifiability provided by Theorem \ref{th:exp_max}, to exclude settings where the bound is trivially equal to the number of nodes $n$.
Algorithm \ref{alg:ICA} formalizes the incremental crossing arrangement design, used to determine the binary encodings of the identifiable nodes.

As we consider $m$ paths, the node encodings will be sequences of $m$ bits in $\mathcal{B} \triangleq \{0,1 \}^m$.
We also denote with $\mathcal{B}|_i \subset \mathcal{B}$ the set of $m$-digits binary encodings  having a 1 in the $i$-th position, 
i.e., $\mathcal{B}|_i= \{ b \in \mathcal{B}\ s.t.\ b|_i=1\}$.
The nodes corresponding to encodings of $\mathcal{B}|_i$ will be monitored (at least) by path $p_i$.
Moreover, we denote with $\mathcal{B}(k) \subset \mathcal{B}$ the set of all binary encodings having exactly $k$ digits equal to 1, therefore
$\mathcal{B}(k) \triangleq \{ b \in \mathcal{B}\ s.t.\  \sum_{i=1}^m b|_i= k\}$. The nodes corresponding to encodings in $\mathcal{B}(k)$ have crossing number equal to $k$.

Finally, given a generic set of binary encodings $B \subseteq \mathcal{B}$, we  denote with $\ell_i(B)$ the number of encodings of $B$ having a one in the $i$-th position:
 $\ell_i(B) \triangleq |B\cap\mathcal{B}|_i|$. The value of $\ell_i(B)$ represents the length of a path $p_i$ traversing all the nodes  in $B \cap \mathcal{B}|_i$, exactly once.

{\color{black}

Without loss of generality, we consider paths of balanced length, i.e. 
 we set the length $d_i$ of path $p_i$ to a value $d_i \in \{ \lfloor \bar{d} \rfloor, \lfloor \bar{d} \rfloor +1 \}$  ({\bf lines  \ref{ICA:2} - \ref{ICA:d2d}}).
%

The incremental crossing arrangement approach incrementally generates the solution set 
$B_V$ by incuding all the encodings of 
$\mathcal{B}(i)$, $i=1, \ldots, i_\texttt{max}$ corresponding to nodes with crossing number lower than or equal to $i_{\texttt{max}}$. 
It then considers some encodings with $(i_{\texttt{max}}+1)$ digits equal to one. For this purpose it generates a family $\mathcal{F}$ of subsets in $\mathcal{B}({i_{\texttt{max}}+1})$, i.e., $\mathcal{F} \subseteq 2^{\mathcal{B}({i_{\texttt{max}}+1})}$ ({\bf line \ref{ICA:F}}) whose elements $B$ are  such that $\ell_k(B\cup B_V) \leq d_k$.
The algorithm then looks for a maximal cardinality set $B^*$ in the family $\mathcal{F}$ and adds it to the solution $B_V$, s.t. 
$B_V= \cup_{k=1}^{i_\texttt{max}} \mathcal {B}(k) \cup B^*$.
Notice that the maximality of the cardinality of $B^*$ implies that no encoding with $(i_\texttt{max}+1)$ digits equal to one can be added to the set $B_V$ without violating the path length constraint $\ell_k(B_V) \leq d_k$ for some path $k=1, \ldots,m$, or without removing at least one encoding already in $B_V$.

The procedure described so far is sufficient to produce a network topology and related paths, meeting the bound of Theorem \ref{th:exp_max}, with $m$ paths of average length lower than or equal to $\bar{d}$.
In the produced topology, there can be values of $k \in \{1, \ldots, m \}$ for which 
$\ell_k(B_V) < d_k$ and, more precisely, given the balanced path length,
$\ell_k(B_V) = d_k -1$, corresponding to paths longer than strictly necessary to meet the bound of Theorem \ref{th:exp_max}, i.e. overlength paths.
Overlength paths cannot traverse nodes with the same encoding without compromising the achievement of maximum identifiability. Therefore, to meet the bound with average path length exactly equal to $\bar{d}$, we proceed as follows, with a procedure that we call {\em Path Completion}.

Let $S \subset \{1, \ldots, m \}$ be the set of overlength path indexes, namely 
$S \triangleq \{k,\ s.t.\  \ell_k(B_V) = d_k -1 \}$.
It holds 
\begin{math}
|S|=\left[(N_{\texttt{max}} - \sum_{i=1}^{i_\texttt{max}}
i\cdot {{m} \choose i })\mod (i_\texttt{max}+1)\right]
\end{math}, 
hence the number of overlength  paths is lower than or equal to $i_\texttt{max}$.

We choose an encoding $b' \in B_V \cap \mathcal{B}(i_\texttt{max}+1-|S|)$ such that $b'|_k=0, \forall k \in S$, 
and such that $\left(\bigvee_{k \in S} \mathbf{e_k} \vee b' \right)\notin B_V$, where $\mathbf{e}_k$ is an $m$-dimensional identity vector with all zeroes but a one in the $k$-th position\footnote{We can always find an encoding $b'$ with the 
described properties  because $B_V$ contains all the 
encodings of $\mathcal{B}(i_\texttt{max}+1 - |S|)$ 
and not all the encodings $b$ of the set 
$ \mathcal{B}(i_\texttt{max}+1)$ for which $b|_i=1, \forall i \in S$.}.
Then we remove  $b'$ from the solution set $B_V$ and replace it with $b'' \triangleq \bigvee_{k \in S} \mathbf{e_k} \vee b'$, i.e., with  a new encoding 
$b''$ such that $b''|_k=1, \forall k \in S,$ and $b''|_k=b'|_k$ otherwise.

{\color{black}
\begin{figure}[h!]
	\centering

 \begin{tikzpicture}
    \tikzstyle{every node}=[draw,circle,fill=white,minimum size=4pt,
                            inner sep=0pt]

    \draw (0,0) node (0100) [label=left:\footnotesize{$\LD 0100 \RD\  $}] {}
        -- ++(300:1cm) node (1100) [label=left:\footnotesize{$\LD 1100 \RD\  $}] {}
        -- ++(60:1cm) node (1000) [label=right:\footnotesize{$\ \LD 1000 \RD\  $}] {}
      ;

    \draw (3.8cm,0) node (0010) [label=left:\footnotesize{$\LD 0010 \RD\  $}] {}
        -- ++(300:1cm) node (0011) [label=right:\footnotesize{$\ \LD 0011 \RD\  $}] {}
        -- ++(60:1cm) node (0001) [label=right:\footnotesize{$\ \LD 0001 \RD\  $}] {}
      ;

 \draw (1100) -- ++(340:2cm) node (1010) [label=left:{\vspace{1cm}\footnotesize{$\LD 1010 \RD\ $}}] {}
   -- ++(0011)  {};

  \node (1101) [below of = 1010, node distance=1cm, label=left:{\footnotesize{$\ \LD 0101 \RD\  $} }]{};
  

\draw [->,cyan,dashed] (0100) to [out=270,in=140] (1100);
\draw [->,cyan,dashed] (1100) to [out=280,in=160] (1101);

\draw  (1100) to (1101);
\draw (0011) to (1101);

\draw [->,blue,densely dotted] (1000) to [out=270,in=40] (1100);
\draw [->,blue,densely dotted] (1100) to [out=350,in=110] (1010);


\draw [->,green] (0010) to [out=270,in=140] (0011);
\draw [->,green] (0011) to [out=190,in=70] (1010);

\draw [->,purple,dashdotted] (0001) to [out=270,in=40] (0011);
\draw [->,purple,dashdotted] (0011) to [out=260,in=20] (1101);

\draw[color=blue,densely dotted] (-2cm, -1.6cm) --  (-1.3cm, -1.6cm) node[draw=none,fill=none] (pippo) [label=right:{{\color{black}$p_1$}}]{};
\draw[cyan,dashed] (-2cm, -1.9cm) -- (-1.3cm, -1.9cm) node[draw=none,fill=none] (pippo) [label=right:{{\color{black}$p_2$}}]{};
\draw[color=green] (-2cm, -2.2cm) -- (-1.3cm, -2.2cm) node[draw=none,fill=none] (pippo) [label=right:{{\color{black}$p_3$}}]{};
\draw[color=purple,dashdotted] (-2cm, -2.50cm) -- (-1.3cm, -2.50cm) node[draw=none,fill=none] (pippo) [label=right:{{\color{black}$p_4$}}]{};

\end{tikzpicture}

 \caption{ICA execution on Example A.
 } 
 \label{fig:ICA}
\end{figure}
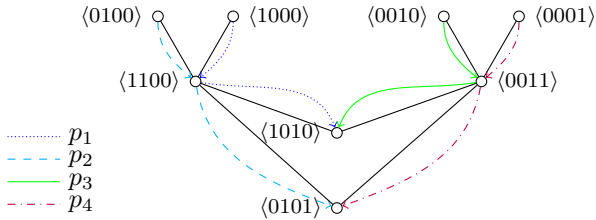
}

{\color{black}
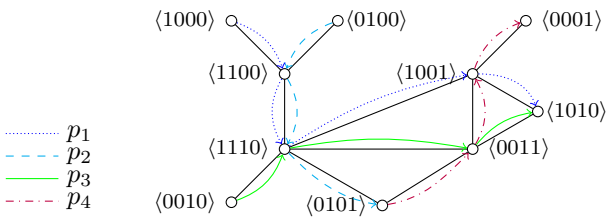
\begin{figure}[h!]
	\centering


 \begin{tikzpicture}
    \tikzstyle{every node}=[draw,circle,fill=white,minimum size=4pt,
                            inner sep=0pt]

    \draw (0,0) node (1000) [label=left:\footnotesize{$\LD 1000 \RD\  $}] {}
        -- ++(315:1cm) node (1100) [label=left:\footnotesize{$\LD 1100 \RD\  $}] {}
        -- ++(270:1cm) node (1110) [label=left:\footnotesize{$\ \LD 1110 \RD\  $}] {}
           -- ++(0:2.5cm) node (0011) [label=right:\footnotesize{$\ \LD 0011 \RD\  $}] {}
            -- ++(90:1cm) node (1001) [label=left:\footnotesize{$\ \LD 1001 \RD\  $}] {}
             -- ++(45:1cm) node (0001) [label=right:\footnotesize{$\ \LD 0001 \RD\  $}] {}
      ;

 \draw (1100) -- ++(45:1cm) node (0100) [label=right:{\vspace{1cm}\footnotesize{$\LD 0100 \RD\ $}}] {};
 \draw (1001) to (1110);

 \draw (1110) -- ++(225:1cm) node (0010) [label=left:\footnotesize{$\LD 0010\RD\  $}] {} {};
 
  \draw (1110) -- ++(330:  1.5cm) node (0101) [label=left:\footnotesize{$\LD 0101\RD\  $}] {} {};
\draw (0101) to (0011);

 \draw (0011) -- ++(30:1cm) node (1010) [label=right:\footnotesize{$\LD 1010 \RD\  $}] {} {};
 \draw  (1001) to (1010);

%
\draw [->,cyan,dashed] (0100) to [out=200,in=70] (1100);
\draw [->,cyan,dashed] (1100) to [out=300,in=60] (1110);
\draw [->,cyan,dashed] (1110) to [out=-45,in=170] (0101);
%
%
\draw [->,blue,densely dotted] (1000) to [out=340,in=110] (1100);
\draw [->,blue,densely dotted] (1100) to [out=240,in=120] (1110);
\draw [->,blue,densely dotted] (1110) to [out=33,in=190] (1001);
\draw [->,blue,densely dotted] (1001) to [out=0,in=110] (1010);
\draw [->,green] (0010) to [out=20,in=250] (1110);
\draw [->,green] (1110) to [out=10,in=170] (0011);
\draw [->,green] (0011) to [out=50,in=190] (1010);
\draw [->,purple,dashdotted] (0101) to [out=10,in=225] (0011);
\draw [->,purple,dashdotted] (0011) to [out=60,in=-60] (1001);
\draw [->,purple,dashdotted] (1001) to [out=70,in=200] (0001);
\draw[color=blue,densely dotted] (-3cm, -1.5cm) --  (-2.3cm, -1.5cm) node[draw=none,fill=none] (pippo) [label=right:{{\color{black}$p_1$}}]{};
\draw[cyan,dashed] (-3cm, -1.8cm) -- (-2.3cm, -1.8cm) node[draw=none,fill=none] (pippo) [label=right:{{\color{black}$p_2$}}]{};
\draw[color=green] (-3cm, -2.1cm) -- (-2.3cm, -2.1cm) node[draw=none,fill=none] (pippo) [label=right:{{\color{black}$p_3$}}]{};
\draw[color=purple,dashdotted] (-3cm, -2.40cm) -- (-2.3cm, -2.40cm) node[draw=none,fill=none] (pippo) [label=right:{{\color{black}$p_4$}}]{};

\end{tikzpicture}

 \caption{ICA execution on Example B.
 } \label{fig:ICA_B}
\end{figure}
}

\begin{algorithm}
\label{alg:ICA}
{\footnotesize{\color{black}
\SetAlgoCaptionSeparator{:}
\KwIn{$m$ and 
$\bar{d}$.}
\KwOut{A set of encodings $B_V$ which can be mapped to a topology graph $G=(V,E)$, with $m$ paths with average length $\bar{d}$, such that $\psi^{*{\texttt{AR}}}(m,\bar{d})$ corresponding nodes are identifiable.}

\begin{algorithmic}[1]
 \STATE Calculate $N_\texttt{max}$ and $i_{\texttt{max}}$ according to Theorem \ref{th:exp_max}, and $\psi^{*{\texttt{AR}}}(m, \bar{d})\triangleq
 \hspace{-.05em}\sum_{i=1}^{i_{\texttt{max}}} \hspace{-.25em} {m \choose i} \hspace{-.25em}+\hspace{-.25em}
\left\lfloor\hspace{-.25em}
\frac{N_{\texttt{max}} - \sum_{i=1}^{i_{\texttt{max}}} i \cdot {m \choose i}   }{i_{\texttt{max}}+1}
\hspace{-.25em}\right\rfloor\hspace{-.25em}$
 \label{ICA:1}
 
 \STATE Calculate $\ m_1
   \triangleq m \cdot ( \bar{d}-\lfloor \bar{d} \rfloor)$ \label{ICA:2};\\

   \STATE {\bf For} $i=1, \ldots , m_1$ {\bf do} set $d_i=\lfloor \bar{d} \rfloor+1$ 
   \label{ICA:d}\\

   \STATE   {\bf For} $i=m_1+1, \ldots, m$ {\bf do} set  $d_i=\lfloor \bar{d} \rfloor$  
   \label{ICA:d2d}\\

      \STATE $B_{V} = \emptyset$\\

    \STATE        {\bf For} $i=1, \ldots, i_\texttt{max}$ {\bf do}
         $B_{V}=B_{V} \cup \mathcal{B}(i)$\label{ICA:complete_tuples}\\

   \STATE         Calculate the family $\mathcal{F}$ defined as
         \newline
         $\mathcal{F}\triangleq \{ B:\  
         B \subseteq 
         \mathcal{B}
         (i_{\texttt{max}}+1)
          \wedge \ell_k(B \cup B_V) \in [ d_k-1, d_k],
           \ \forall k \}$\label{ICA:F}\\

   \STATE             Choose $B^* =\arg \max_{B \in \mathcal{F}}  |B|$\\
       \STATE     $B_V=B_V \cup B^*$ \label{ICA:enc_not_final}\\

  \STATE    \If{ $\exists k \in \{ 1, \ldots, m\}$ s.t.  
    $\ell_k(B_V) = d_k-1$} {Perform {\em path completion} and update $B_V$}      
    \label{ICA:enc_final}

  \STATE  Return ${B_V}$
 \end{algorithmic}
\caption{Incremental Crossing Arrangement}
}}
\end{algorithm}

{\em ICA: example A (where path completion is not necessary).} 
Figure \ref{fig:ICA} shows an example of a topology generated by means of incremental crossing arrangement.
We are given $m=4$ and $\bar{d}=3$, and $n$ arbitrarily large (any value larger than 8 works in this example). Applying Algorithm \ref{alg:ICA} we have $N_\texttt{max}=m\cdot \bar{d}=12$, and $i_\texttt{max}=1$. 
We also have 
$\psi^{*\texttt{AR}}=8$.
We set $d_i=3, \forall i \in \{1, \ldots, 4 \}$ ({\bf lines \ref{ICA:2} - \ref{ICA:d2d}}).
According to ICA, we first generate all the encodings of $\mathcal{B}(i_{\texttt{max}})=\mathcal{B}(1)$ and set $B_V=\mathcal{B}(1)=\{1000,0100,0010,0001 \}$ ({\bf line \ref{ICA:complete_tuples}}).
Then we generate some encodings in $\mathcal{B}(2)$ until no other encoding can be added without violating the path length constraint ({\bf line \ref{ICA:enc_not_final}}), obtaining $B_V=\{ 1000,0100,0010,0001,1100,0011,1010,0101\}$, where each encoding corresponds to a node of the graph $G$. 
Then we define the corresponding monitoring paths, by letting path $p_i$ traverse all the nodes whose encoding 
has a 1 in the $i$-th position, in arbitrary order, $\forall i \in \{1, \ldots, m \}$.
Finally, we design the underlying topology by connecting each pair of nodes appearing in a sequence in any of the paths, as shown in Figure  \ref{fig:ICA}.

{\em ICA: example B (with path completion).} 
Figure \ref{fig:ICA_B} shows another example of a topology generated by means of incremental crossing arrangement.
We are given $m=4$ and $\bar{d}=4.25$, and $n$ arbitrarily large (any value larger than 10 works in this example). Applying Algorithm \ref{alg:ICA} we have $N_\texttt{max}=m\cdot \bar{d}=17$, and $i_\texttt{max}=2$. 
We also have 
$\psi^{*\texttt{AR}}=10$.
To meet the requirement on average lenght, we set $d_1=5$, and $d_2=d_3=d_4=4$ ({\bf lines \ref{ICA:2} - \ref{ICA:d2d}}).
According to ICA ({\bf line \ref{ICA:complete_tuples}}), we first generate all the encodings of $\mathcal{B}(1)$ and $\mathcal{B}(2)$ and set $B_V=\{1000,0100,0010,0001,1100, 1010, 1001, {\underline{0110}}, 0101, 0011 \}$.

Finally, we observe that $\ell_1 (B_V) =4<d_1$.
We then perform the path completion procedure ({\bf line \ref{ICA:enc_final}}) and choose one of the encodings $b'$ in $B_V \cap \mathcal{B}(i_\texttt{max}+1-|S|)=\mathcal{B}(2) $ for which  $b'|_1=0$
and $b' \vee  \mathbf{e_1} \notin B_V$.
One encoding that satisfies this condition is 
$b'=0110$. 
We replace $b'$ with $b''=1110$.
We obtain the set of encodings $\{ 1000,0100,0010,0001,1100,1010,1001,{\underline{1110}},0101,0011\}$, each corresponding to a node of the graph $G$. 
Then we define the corresponding monitoring paths, by letting path $p_i$ traverse all the nodes whose encoding 
has a 1 in the $i$-th position, in arbitrary order, $\forall i \in \{1, \ldots, m \}$.
Finally, we design the underlying topology by connecting each pair of nodes appearing in a sequence in any of the paths, obtaining the topology of Figure \ref{fig:ICA_B}.


}

\noindent It is worth observing the following.
\begin{observation}\label{obs:max_cross}
ICA produces a network topology and related monitoring paths such that all nodes  have a crossing number lower than or equal to $(i_\texttt{max}+1)$.

\end{observation}
\subsubsection{Tightness of the bound on identifiability under arbitrary routing}\label{subsubsec:Design - arbitrary routing}
In this section we  show that the bound given by Theorem \ref{th:exp_max} can be achieved tightly  for a specific family of topologies constructed  via ICA.

\begin{prop}[Tightness of Theorem \ref{th:exp_max}]
\label{prop:tightness_arbitrary_routing}

For any $m\in \mathbb{Z}^+$ (positive integer) and $\bar{d}>0$, there exists a set $P$ of $m$ monitoring paths with average length $\bar{d}$, such that the number of nodes identifiable by monitoring $P$ equals the bound given in Theorem  \ref{th:exp_max}:

\begin{align*}
 \psi^{*{\texttt{AR}}}(m,\bar{d}\hspace{+0.05em})
= \sum_{i=1}^{i_{\texttt{max}}} {m \choose i} +
\left\lfloor
\frac{N_{\texttt{max}} - \sum_{i=1}^{i_{\texttt{max}}} i \cdot {m \choose i}   }{i_{\texttt{max}}+1}
\right\rfloor.
\end{align*}

\end{prop}

\begin{proof}

We recall that 
the ICA technique builds a topology by creating nodes with unique encodings, in increasing order of crossing number, up to $(i_\texttt{max}+1)$. 
To prove the proposition, we need to show that the number of identifiable nodes is equal to the one provided by the bound of Theorem \ref{th:exp_max}.
ICA initially generates all the encodings of  
$\mathcal{B}(i)$, for $i=1, \ldots, i_\texttt{max}$.
As a consequence, notice that each path will traverse at least 
$d(i_\texttt{max})\triangleq
 \sum_{i=0}^{i_\texttt{max}-1}{{m-1}\choose i}$ identifiable nodes. In fact, the encodings of the nodes of $\mathcal{I}(p_i)$ (identifiable nodes traversed by path $p_i$), must have a "1" in the $i$-th position. Therefore the number of distinct encodings corresponding to nodes of $\mathcal{I}(p_i)$ is at least equal to the number of binary sequences of $(m-1)$ elements, with up to $(i_\texttt{max}-1)$ ones, which is 
$d(i_\texttt{max})$. 


%
Under incremental crossing arrangement, each path also traverses  other nodes with crossing number equal to $(i_\texttt{max}+1)$.
Each of these nodes will appear in exactly 
$(i_\texttt{max}+1)$ paths. 
 The number of such nodes  is therefore given
by 
 $\left\lfloor \frac{\sum_{k=1}^m (d_k - d(i_\texttt{max}))}{( i_{\texttt{max}}+1)}\right\rfloor$.

In conclusion, with this construction, ICA generates the following number of node encodings:
\begin{itemize}
\item $m \choose i$ encodings corresponding to nodes with crossing number  equal to $i$,  for $i=1,\ldots i_\texttt{max}$, and
\item   $\left\lfloor \frac{\sum_{k=1}^m (d_k - d(i_\texttt{max}))}{( i_{\texttt{max}}+1)}\right\rfloor$ encodings corresponding to nodes with crossing number equal to $(i_\texttt{max}+1)$. \end{itemize}

The number of generated encodings does not change if ICA applies the path completion procedure, which consists in a replacement of an encoding $b' \in \cup_{i=1}^{i_{\texttt{max}}} \mathcal{B}(i)$
with an encoding $b'' \in \mathcal{B}(i_\texttt{max}+1)$.
In both cases, ICA constructs the set $B_V$ in a way that each  encoding corresponds to a unique node, and the nodes are traversed by paths of average length $\bar{d}$, guaranteeing identifiability of all the nodes corresponding to the generated encodings.

In order to show that the number of identifiable nodes is equal to the one provided by the bound of Theorem \ref{th:exp_max},   we need to prove that $\left\lfloor\frac{\sum_{k=1}^m (d_k - d(i_\texttt{max}))}{( i_{\texttt{max}}+1)} \right\rfloor= 
\left\lfloor
\frac{N_{\texttt{max}} - \sum_{i=1}^{i_{\texttt{max}}} i \cdot {m \choose i}}{(i_{\texttt{max}}+1)}
\right\rfloor$, which  holds because $ \sum_{k=1}^m d_k=m\cdot \bar{d}=N_{\texttt{max}}$, and $m \cdot d(i_\texttt{max}) = m \cdot\sum_{i=0}^{i_{\texttt{max}}-1} {{m-1} \choose i}=\sum_{i=1}^{i_{\texttt{max}}} i \cdot {m \choose i}$, which can easily be proven by expanding the binomial coefficients.


\end{proof}

Notice that Proposition \ref{prop:tightness_arbitrary_routing}
requires $\bar{d}\leq 2^{m-1}$ as having longer paths would require
at least a path to traverse different nodes with duplicate encodings, losing identifiability with respect to the bound value.
 
{\color{black}{
While Proposition \ref{prop:tightness_arbitrary_routing}
gives a characterization of sufficient conditions for building a network topology achieving the bound,
we 
note that there exist topologies that do not meet the conditions, but still achieve the bound. We leave the characterization of necessary conditions for achieving the bound defined in Theorem IV.1 to future work.

\subsection{Consistent routing}
As we have seen in Theorem \ref{th:exp_max}, given a number of monitoring paths, the number of  identifiable nodes can be exponential  in the number of paths.
Nevertheless the bound of Theorem \ref{th:exp_max} is achieved  only when the routing scheme allows 
 paths to traverse arbitrary sequences of nodes.

If routing needs to meet additional requirements,
%
%
%
the theoretical bound given by Theorem \ref{th:exp_max} can be reduced.

We now consider the impact of the routing scheme on the  identifiability of nodes via Boolean tomography.
\subsubsection{Identifiability bound}
{\color{black}{
In the sequel, we assume that paths  satisfy the following property of {\em routing consistency}.
\begin{definition}
A set of paths $P$ is \emph{consistent} if $\forall p,\: p'\in P$ and any two nodes $u$ and $v$ traversed by both paths (if any), $p$ and $p'$  follow the same sub-path between $u$ and $v$.
\label{def:consistent_routing}
\end{definition}

Figure \ref{fig:ICA_B} is an example of non-consistent routing. Indeed, some monitoring paths  traverse different routes between  the same pair of nodes. For example paths $p_1$ and $p_3$ choose different routes to go from node $1110$ to node $1010$,  across nodes  $1001$ and  $0011$, respectively.
Nevertheless, if $p_1$ followed the same route as $p_3$, through node $0011$,  the node currently having encoding $1001$ would have the  new encoding $0001$, and it would no longer be identifiable due to the simultaneous presence of another node with the same encoding.

An example of consistent routing  of monitoring paths is instead given in Figure \ref{fig:all_1_id_graph_usp}.

{\color{black}We remark that routing consistency
is satisfied by many practical routing protocols, including but not limited to shortest path routing (where ties are broken with a unique deterministic rule).
 Note that routing consistency implies that paths are cycle-free.}

 We define the  {\em path matrix} of  $\hat{p}_i$ as a binary matrix $M(\hat{p}_i)$, in which each  row is the binary encoding of a node on the path, and rows are sorted according to the sequence $\hat{p}_i$. Notice that by definition $M(\hat{p}_i)|_{*,i}$ has only ones, i.e., $M(\hat{p}_i)|_{r,i} =1, \ \forall r$. 



\begin{lemma}\label{le:equal_rows}
Under the assumption of consistent routing, if  any two different rows of the matrix $M(\hat{p}_i)$ are equal, then the corresponding nodes are not 1-identifiable.
\end{lemma}
\begin{proof}
Under consistent routing, the path $\hat{p}_i$ cannot contain any cycle, so every row of  $M(\hat{p}_i)$ corresponds to a different node.
If two different nodes have the same binary encoding,  by Lemma \ref{def:identifiability},  the two nodes are not identifiable.
\end{proof}


\begin{definition}
A column $M(\hat{p})|_{*,k}$ ($k=1,\ldots,m$) of a path matrix $M(\hat{p})$ has {\em consecutive ones} if all the ``1''s appear in consecutive rows, i.e., for any two rows $i$ and $j$ ($i<j$), if $M(\hat{p})|_{i,k}=M|_{j,k}=1$, then $M|_{h,k}=1$ for all $i\leq h \leq j$.
\end{definition}
\begin{lemma}\label{le:c1p}
Under the assumption of consistent routing, all the columns in all the path matrices have consecutive ones.
\end{lemma}
\begin{proof}
 The assertion is true for $M(\hat{p}_i)|_{*,i}$ since it contains only ones. Let us consider 
 column $M(\hat{p}_i)|_{*,j}$, with
  $j \not = i $. Assume by contradiction
 that there are two rows $k_1< k_2$ s.t. $M(\hat{p}_i)|_{k_1,j}=M(\hat{p}_i)|_{k_2,j}=1$ but
 there is a row $h$ with $k_1 <h<k_2$ for which $M(\hat{p}_i)|_{h,i}=0$. 
 Let $v_1$, $v_2$, and $v_h$ be the nodes with encodings $M(\hat{p}_i)|_{k_1,*}$, $M(\hat{p}_i)|_{k_2,*}$, and $M(\hat{p}_i)|_{h,*}$, respectively.
Then the paths $\hat{p}_i$ and $\hat{p}_j$  traverse both nodes $v_1$ and $v_2$ following different paths, of which only $\hat{p}_i$ traverses node $v_h$, in contradiction with consistent routing.
\end{proof}

\begin{lemma}\label{le:numero}
\color{black}{
Given $m=|P|>1$ consistent routing paths, each path $p_i$ having length $d_i$, 
the maximum number of different encodings in the rows of $M(\hat{p}_i)$ is upper-bounded by $\min \{d_i;  2 \cdot (m-1)\}$.}
\end{lemma}
\begin{proof} While the number of different encodings appearing in the rows of $M(\Hat{p}_i)$ is trivially {\color{black}bounded by $d_i$,} it can even be lower. 
By considering each column of $M(\hat{p}_i)$ separately we observe the following.
First, column $M(\hat{p}_i)|_{*,i}$ contains only ones. Second,
for any column $M(\hat{p}_i)|_{*,j}$ with $j \neq i$, it holds, by Lemma \ref{le:c1p}, that
it has a consecutive ones.

We say that column $k$ has a {\em flip} in row $r$ if $M(\hat{p}_i)|_{r-1,k}\neq M(\hat{p}_i)|_{r,k}$.
Due to Lemma \ref{le:c1p} any column of $M(\hat{p}_i)$ can have up to two flips or it would create a fragmented sequence of ones, violating Lemma \ref{le:c1p}.
In fact, if the column starts with a 0 in the first row, it can flip from 0 to 1 in row $r_1$ and then back in row $r_2$, with $r_2 >r_1$, but if it flips from 1 to 0 it can not flip back in a successive column.
If instead the column starts with a 1 in the first row, it can only flip once.
In order to have a change in the encoding contained in any two successive rows $r-1$ and $r$ of the matrix $M(\hat{p}_i)$, i.e., $M(\hat{p}_i)|_{r-1,*} \neq M(\hat{p}_i)|_{r,*}$,
there must be at least a column  that flips in $r$.
The  number of columns that can flip is $m-1$ and each of them can  flip at most two times.
The number of different  rows that can be observed in $M(\hat{p}_i)$  is therefore upper-bounded by the smallest between the path length $d_i$ and $2\cdot(m-1)$.
\end{proof}
}

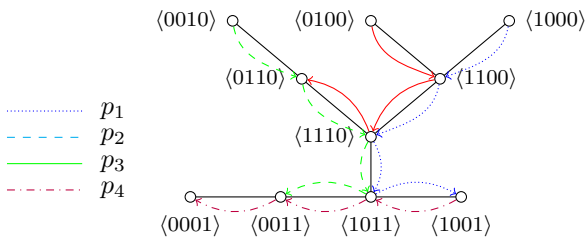
\begin{figure}[h]
	\centering

\begin{tikzpicture}

\tikzset{pallino/.style={draw,circle,fill=white,minimum size=4pt,
                            inner sep=0pt}}
\tikzset{mylab/.style={black, minimum size=0.5cm}}

  \draw (0,0) node (0010) [pallino,label=left:\footnotesize{$\LD 0010 \RD$}] {}
        -- ++(320:1.2cm) node (0110) [pallino,label=left:\footnotesize{$\LD 0110 \RD$}] {}
       -- ++(320:1.2cm) node (1110) [pallino,label=left:\footnotesize{$\ \LD 1110 \RD$}] {}
           -- ++(270:0.8cm) node (1011) [pallino]{}
           --++ (180:1.2cm) node (0011) [pallino]{}
            --++ (180:1.2cm) node (0001) [pallino]{};
    
 \draw (1110) -- ++(40:1.2cm) node (1100) [pallino,label= right:\footnotesize{$\LD 1100 \RD$}] {}   
 --++ (40:1.2cm) node (1000) [pallino,label= right:\footnotesize{$\LD 1000 \RD$}] {}  ; 
 
  \draw (1100) -- ++(140:1.2cm) node (0100) [pallino,label= left:\footnotesize{$\LD 0100 \RD$}] {}  ; 
\draw (1011) -- ++(360:1.2cm) node (1001) [pallino] {}  ;

        \node [mylab,below of = 1011, node distance=0.15in] (blank) {\footnotesize{$\ \LD 1011 \RD\  $} };
        \node [mylab,below of = 0011, node distance=0.15in] (blank) {\footnotesize{$\ \LD 0011 \RD\  $} };
        \node [mylab,below of = 0001, node distance=0.15in] (blank) {\footnotesize{$\ \LD 0001 \RD\  $} };
         \node [mylab,below of = 1001, node distance=0.15in] (blank) {\footnotesize{$\ \LD 1001 \RD\  $} };

\draw [->,green,dashed] (0010) to [out=280,in=160] (0110);
\draw [->,green,dashed] (0110) to [out=280,in=160] (1110);
\draw [->,red] (0100) to [out=280,in=160] (1100);

\draw [->,blue,densely dotted] (1000) to [out=260,in=20] (1100);
\draw [->,blue,densely dotted] (1100) to [out=260,in=20] (1110);
\draw [->,red] (1100) to [out=190,in=70] (1110);
\draw [->,red] (1110) to [out=110,in=350] (0110);
\draw [->,green,dashed] (1110) to [out=240,in=120] (1011);
\draw [->,blue,densely dotted] (1110) to [out=300,in=60] (1011);
\draw [->,blue,densely dotted] (1011) to [out=30,in=150] (1001);
\draw [->,green,dashed] (1011) to [out=150,in=30] (0011);

\draw [->,purple,dashdotted] (1001) to [out=210,in=330] (1011);
\draw [->,purple,dashdotted] (1011) to [out=210,in=330] (0011);
\draw [->,purple,dashdotted] (0011) to [out=210,in=330] (0001);

\draw[color=blue,densely dotted] (-3cm, -1.2cm) --  (-2cm, -1.2cm) node[draw=none,fill=none] (pippo) [label=right:{{\color{black}$p_1$}}]{};
\draw[cyan,dashed] (-3cm, -1.55cm) -- (-2cm, -1.55cm) node[draw=none,fill=none] (pippo) [label=right:{{\color{black}$p_2$}}]{};
\draw[color=green] (-3cm, -1.9cm) -- (-2cm, -1.9cm) node[draw=none,fill=none] (pippo) [label=right:{{\color{black}$p_3$}}]{};
\draw[color=purple,dashdotted] (-3cm, -2.25cm) -- (-2cm, -2.25cm) node[draw=none,fill=none] (pippo) [label=right:{{\color{black}$p_4$}}]{};

	\end{tikzpicture}

 \caption{Consistent routing paths identifying all nodes of the network.} \label{fig:all_1_id_graph_usp}
\end{figure}

}}
{\color{black}{

For example, the matrices of the paths of Figure \ref{fig:all_1_id_graph_usp}  
have columns with consecutive ones and each column flips at most twice, so the number of different rows is lower  than, or equal to $2 \cdot(m-1)=6$.
For instance,  $M(\hat{p}_3)$ is:

{\footnotesize
\begin{equation}
\nonumber
M(\hat{p}_3)=\
\bbordermatrix{
 \mbox{flips}  &b_1 & b_2  & b_{3} &  b_{4} \cr
0 &0 &0  &1  &0 \cr
1 &0 &1  &1  &0 \cr
2 &1 &1  &1  &0  \cr
3 &1 &0  &1  &1 \cr
4 &0 &0  &1  &1 \cr
}
\end{equation}
}
}

We now give an upper bound on the number of identifiable nodes under consistent routing.

{\color{black}{

\begin{theorem} [Identifiability with consistent routing] \label{th:bound_consistent_routing}
 Given $n$ nodes, and a set $P$ of $m>1$ consistent routing paths, with average path length $\bar{d}$, the maximum number of identifiable nodes $\psiCPR$, for any $G$ and any location of the path endpoints,  is upper-bounded as in Theorem \ref{th:exp_max}, 
 \begin{align}
 \hspace{-.5em}
 \psiCPR\hspace{-.1em}(\hspace{-.05em}m,n,\bar{d}\hspace{+0.05em})
&\hspace{-.05em}\leq \hspace{-.05em} \min \hspace{-.15em}\left\{\hspace{-.05em}\sum_{i=1}^{i_{\texttt{max}}} \hspace{-.25em} {m \choose i} \hspace{-.25em}+\hspace{-.25em}
\left\lfloor\hspace{-.25em}
\frac{N_{\texttt{max}} - \sum_{i=1}^{i_{\texttt{max}}} i \cdot {m \choose i}   }{i_{\texttt{max}}+1}
\hspace{-.25em}\right\rfloor\hspace{-.25em}; n \hspace{-.10em}\right\},
\nonumber
\end{align}
\noindent
where
$i_{\texttt{max}}= \max \{
k \ | \ \sum_{i=1}^{k} i \cdot {m \choose i} \leq N_{\texttt{max}} \}$, 
  except that $N_{\texttt{max}}$ is now defined as follows:
 \begin{center}
     \begin{math}
     N_{\texttt{max}}=m \cdot \min 
     \{\bar{d}\ ; 2 \cdot (m-1)\}.
     \end{math}
 \end{center}\end{theorem}
 
\begin{proof}
The proof is analogous to the one of Theorem \ref{th:exp_max}, as again we want to minimize the number of ones in the encodings of the nodes in order to avoid repetitions. The difference with the arbitrary routing case lies in the value of $N_{\texttt{max}}$, that now is the sum, extended to all paths, of the bound shown in Lemma \ref{le:numero}. 
 \end{proof}

As we did in the case of arbitrary routing, 
we focus on the situation in which there is an upper  bound on the length of monitoring paths, but the individual path length is not fixed, nor is the average path length. In this case, we have the following variation of the bound due to the fact that 
$$ \bar{d} \leq \max_i \{d_i\} \leq d_{\texttt{max}}.$$

\begin{corollary}
[Identifiability under consistent routing, and bounded maximum path length]
\label{corollary:consistent_routing}
 Given a network and a set $P$ of $m>1$ consistent routing paths with maximum length $d_\texttt{max}$, the maximum number of identifiable nodes in the network is upper-bounded as in Theorem \ref{th:exp_max}, except that $N_{\texttt{max}}$ is now defined as: 
\noindent
$N_{\texttt{max}} =  m \cdot \min \{ d_\texttt{max}; \  2\cdot(m-1)\}$.
\end{corollary}

} 
}
\subsubsection{Tightness of the bound and design insights}\label{subsubsec:Design - consistent routing}

\begin{figure}[t]

	\centering

\begin{tikzpicture}

\tikzset{pallino/.style={draw,circle,fill=white,minimum size=4pt,
                            inner sep=0pt}}

\tikzset{pallinorosso/.style={draw,circle,fill=red,minimum size=4pt,
                            inner sep=0pt}}
\tikzset{mylab/.style={black, minimum size=0.5cm}}


  \draw (0,0) node (1) [pallino] {}
        -- ++(0:0.8cm) node (2) [pallino] {}
        -- ++(0:0.8cm) node (3) [pallino] {}
        -- ++(0:0.8cm) node (4) [pallino] {}
        -- ++(0:0.8cm) node (5) [pallino] {}
        -- ++(0:0.8cm) node (6) [pallino] {}
        -- ++(0:0.8cm) node (7) [pallino] {}
        -- ++(0:0.8cm) node (8) [pallino] {};

 \draw (0.8,0.6) node (9) [pallino] {}
        -- ++(0:0.8cm) node (10) [pallino] {}
        -- ++(0:0.8cm) node (11) [pallino] {}
        -- ++(0:0.8cm) node (12) [pallino] {}
        -- ++(0:0.8cm) node (13) [pallino] {}
        -- ++(0:0.8cm) node (14) [pallino] {}
        -- ++(0:0.8cm) node (15) [pallino] {};      
        
 \draw (1.6,1.2) node (16) [pallino] {}
        -- ++(0:0.8cm) node (17) [pallino] {}
        -- ++(0:0.8cm) node (18) [pallino] {}
        -- ++(0:0.8cm) node (19) [pallino] {}
        -- ++(0:0.8cm) node (20) [pallino] {}
        -- ++(0:0.8cm) node (21) [pallino] {};   
        
 \draw (2.4,1.8) node (22) [pallino] {}
        -- ++(0:0.8cm) node (23) [pallino] {}
        -- ++(0:0.8cm) node (24) [pallino] {}
        -- ++(0:0.8cm) node (25) [pallino] {}
        -- ++(0:0.8cm) node (26) [pallino] {};  
        
  \draw (3.2,2.4) node (27) [pallino] {}
        -- ++(0:0.8cm) node (28) [pallino] {}
        -- ++(0:0.8cm) node (29) [pallino] {}
        -- ++(0:0.8cm) node (30) [pallino] {};       
        
 \draw (4.0,3.0) node (31) [pallino] {}
        -- ++(0:0.8cm) node (32) [pallino] {}
        -- ++(0:0.8cm) node (33) [pallino] {};  
 
 \draw (4.8,3.6) node (34) [pallino] {}
        -- ++(0:0.8cm) node (35) [pallino] {};     
        
\draw (5.6,4.2) node (36) [pallino] {};   


\draw [-] (2) -- (9);
\draw [-] (3) --(10)--(16) ;
\draw [-] (4) -- (11) -- (17)--  (22);
\draw [-] (5) --(12) --(18) -- (23) -- (27);
\draw [-] (6) -- (13) -- (19) -- (24) -- (28) --  (31);
\draw [-] (7) -- (14) -- (20) -- (25) -- (29) -- (32) -- (34);

\draw [-] (1) --(9) -- (16) -- (22) -- (27) -- (31) -- (34) --  (36);


\draw [->,red] (8) to [out=150,in=30] (7);      
\draw [->,red] (7) to [out=150,in=30] (6);     
\draw [->,red] (6) to [out=150,in=30] (5);    
\draw [->,red] (5) to [out=150,in=30] (4);   
\draw [->,red] (4) to [out=150,in=30] (3);  
\draw [->,red] (3) to [out=150,in=30] (2); 
\draw [->,red] (2) to [out=150,in=30] (1);

\draw [->,green] (15) to [out=150,in=30] (14);      
\draw [->,green] (14) to [out=150,in=30] (13);     
\draw [->,green] (13) to [out=150,in=30] (12);    
\draw [->,green] (12) to [out=150,in=30] (11);   
\draw [->,green] (11) to [out=150,in=30] (10);  
\draw [->,green] (10) to [out=150,in=30] (9); 
\draw [->,green] (9) to [out=185,in=60] (1); 

\draw [->,blue] (21) to [out=150,in=30] (20);      
\draw [->,blue] (20) to [out=150,in=30] (19);     
\draw [->,blue] (19) to [out=150,in=30] (18);    
\draw [->,blue] (18) to [out=150,in=30] (17);   
\draw [->,blue] (17) to [out=150,in=30] (16);  
\draw [->,blue] (16) to [out=185,in=60] (9); 
\draw [->,blue] (9) to [out=235,in=130] (2); 

\draw [->,magenta] (26) to [out=150,in=30] (25);      
\draw [->,magenta] (25) to [out=150,in=30] (24);     
\draw [->,magenta] (24) to [out=150,in=30] (23);    
\draw [->,magenta] (23) to [out=150,in=30] (22);   
\draw [->,magenta] (22) to [out=185,in=60]  (16);  
\draw [->,magenta] (16) to [out=235,in=130] (10); 
\draw [->,magenta] (10) to [out=235,in=130] (3); 

\draw [->,orange] (30) to [out=150,in=30] (29);      
\draw [->,orange] (29) to [out=150,in=30] (28);     
\draw [->,orange] (28) to [out=150,in=30] (27);    
\draw [->,orange] (27) to [out=185,in=60] (22);   
\draw [->,orange] (22) to [out=235,in=130]  (17);  
\draw [->,orange] (17) to [out=235,in=130] (11); 
\draw [->,orange] (11) to [out=235,in=130] (4); 

\draw [->,teal] (33) to [out=150,in=30] (32);      
\draw [->,teal] (32) to [out=150,in=30] (31);     
\draw [->,teal] (31) to [out=185,in=60] (27);    
\draw [->,teal] (27) to [out=235,in=130]  (23);   
\draw [->,teal] (23) to [out=235,in=130]  (18);  
\draw [->,teal] (18) to [out=235,in=130] (12); 
\draw [->,teal] (12) to [out=235,in=130] (5); 

\draw [->,brown] (35) to [out=150,in=30] (34);      
\draw [->,brown] (34) to [out=185,in=60](31);     
\draw [->,brown] (31) to [out=235,in=130] (28);    
\draw [->,brown] (28) to [out=235,in=130]  (24);   
\draw [->,brown] (24) to [out=235,in=130]  (19);  
\draw [->,brown] (19) to [out=235,in=130] (13); 
\draw [->,brown] (13) to [out=235,in=130] (6); 

\draw [->,violet] (36) to [out=185,in=60] (34);     
\draw [->,violet] (34) to [out=235,in=130](32);     
\draw [->,violet] (32) to [out=235,in=130] (29);    
\draw [->,violet] (29) to [out=235,in=130]  (25);   
\draw [->,violet] (25) to [out=235,in=130]  (20);  
\draw [->,violet] (20) to [out=235,in=130] (14); 
\draw [->,violet] (14) to [out=235,in=130] (7); 

\node [mylab,below of = 1, node distance=0.12in] (blank){\footnotesize{$t_1,t_2$} };
\node [mylab,below of = 2, node distance=0.12in] (blank){\footnotesize{$t_3$} };
\node [mylab,below of = 3, node distance=0.12in] (blank){\footnotesize{$t_4$} };
\node [mylab,below of = 4, node distance=0.12in] (blank){\footnotesize{$t_5$} };
\node [mylab,below of = 5, node distance=0.12in] (blank){\footnotesize{$t_6$} };
\node [mylab,below of = 6, node distance=0.12in] (blank){\footnotesize{$t_7$} };
\node [mylab,below of = 7, node distance=0.12in] (blank){\footnotesize{$t_8$} };

\node [mylab,right of = 36, node distance=0.12in] (blank){\footnotesize{$s_8$} };
\node [mylab,right of = 35, node distance=0.12in] (blank){\footnotesize{$s_7$} };
\node [mylab,right of = 33, node distance=0.12in] (blank){\footnotesize{$s_6$} };
\node [mylab,right of = 30, node distance=0.12in] (blank){\footnotesize{$s_5$} };
\node [mylab,right of = 26, node distance=0.12in] (blank){\footnotesize{$s_4$} };
\node [mylab,right of = 21, node distance=0.12in] (blank){\footnotesize{$s_3$} };
\node [mylab,right of = 15, node distance=0.12in] (blank){\footnotesize{$s_2$} };
\node [mylab,right of = 8, node distance=0.12in] (blank){\footnotesize{$s_1$} };

%
%
	
	\end{tikzpicture}

\vspace{-.3cm}	
	 \caption{An example of half-grid graph} \label{fig:half_grid}
	\vspace{-.5cm}	
\end{figure}
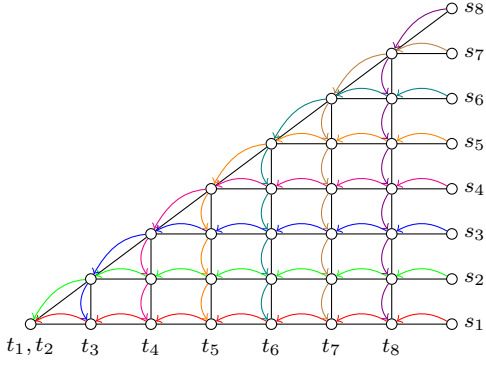 

\begin{figure}[t]
	\centering
\begin{tikzpicture}

\tikzset{pallino/.style={draw,circle,fill=white,minimum size=4pt,
                            inner sep=0pt}}

\tikzset{pallinorosso/.style={draw,circle,fill=red,minimum size=4pt,
                            inner sep=0pt}}
\tikzset{mylab/.style={black, minimum size=0.5cm}}


  \draw (0,0) node (1) [pallino] {};
  \draw (0:0.8cm) node (2) [pallino] {}
        -- ++(0:0.8cm) node (3) [pallino] {}
        -- ++(0:0.8cm) node (4) [pallino] {}
        -- ++(0:0.8cm) node (5) [pallino] {}
        -- ++(0:0.8cm) node (6) [pallino] {}
        -- ++(0:0.8cm) node (7) [pallino] {}
        -- ++(0:0.8cm) node (8) [pallino] {};

\draw (0.2,0.45) node (37) [pallinorosso] {};
\draw [-] (2) -- (37);

 \draw (0.8,0.6) node (9) [pallino] {}
        -- ++(0:0.8cm) node (10) [pallino] {}
        -- ++(0:0.8cm) node (11) [pallino] {}
        -- ++(0:0.8cm) node (12) [pallino] {}
        -- ++(0:0.8cm) node (13) [pallino] {}
        -- ++(0:0.8cm) node (14) [pallino] {}
        -- ++(0:0.8cm) node (15) [pallino] {};      
        
 \draw (1.6,1.2) node (16) [pallino] {}
        -- ++(0:0.8cm) node (17) [pallino] {}
        -- ++(0:0.8cm) node (18) [pallino] {}
        -- ++(0:0.8cm) node (19) [pallino] {}
        -- ++(0:0.8cm) node (20) [pallino] {}
        -- ++(0:0.8cm) node (21) [pallino] {};   
        
 \draw (2.4,1.8) node (22) [pallino] {};
\draw (3.2,1.8)node (23) [pallino] {}
        -- ++(0:0.8cm) node (24) [pallino] {}
        -- ++(0:0.8cm) node (25) [pallino] {}
        -- ++(0:0.8cm) node (26) [pallino] {};  
        
\draw (2.6,2.25) node (38) [pallinorosso] {};
   \draw [-] (23) -- (38);     
        
  \draw (3.2,2.4) node (27) [pallino] {}
        -- ++(0:0.8cm) node (28) [pallino] {}
        -- ++(0:0.8cm) node (29) [pallino] {}
        -- ++(0:0.8cm) node (30) [pallino] {};       
        
 \draw (4.0,3.0) node (31) [pallino] {}
        -- ++(0:0.8cm) node (32) [pallino] {}
        -- ++(0:0.8cm) node (33) [pallino] {};  
 
 \draw (4.8,3.6) node (34) [pallino] {}
        -- ++(0:0.8cm) node (35) [pallino] {};     
        
\draw (5.6,4.2) node (36) [pallino] {};   


\draw [-] (3) --(10)--(16) ;
\draw [-] (4) -- (11) -- (17)--  (22);
\draw [-] (5) --(12) --(18) -- (23);
\draw [-] (6) -- (13) -- (19) -- (24) -- (28) --  (31);
\draw [-] (7) -- (14) -- (20) -- (25) -- (29) -- (32) -- (34);

\draw [-] (1) --(37) -- (9) -- (16) -- (22) -- (38)--(27) -- (31) -- (34) --  (36);


\draw [->,red] (8) to [out=150,in=30] (7);      
\draw [->,red] (7) to [out=150,in=30] (6);     
\draw [->,red] (6) to [out=150,in=30] (5);    
\draw [->,red] (5) to [out=150,in=30] (4);   
\draw [->,red] (4) to [out=150,in=30] (3);  
\draw [->,red] (3) to [out=150,in=30] (2); 
\draw [->,red] (2) to [out=180,in=290] (37); 
\draw [->,red] (37) to [out=270,in=30] (1); 
\draw [->,green] (15) to [out=150,in=30] (14);      
\draw [->,green] (14) to [out=150,in=30] (13);     
\draw [->,green] (13) to [out=150,in=30] (12);    
\draw [->,green] (12) to [out=150,in=30] (11);   
\draw [->,green] (11) to [out=150,in=30] (10);  
\draw [->,green] (10) to [out=150,in=30] (9); 

\draw [->,green] (9) to [out=150,in=60] (37); 
\draw [->,green] (37) to [out=210,in=100] (1); 

\draw [->,blue] (21) to [out=150,in=30] (20);      
\draw [->,blue] (20) to [out=150,in=30] (19);     
\draw [->,blue] (19) to [out=150,in=30] (18);    
\draw [->,blue] (18) to [out=150,in=30] (17);   
\draw [->,blue] (17) to [out=150,in=30] (16);  
\draw [->,blue] (16) to [out=185,in=60] (9); 
\draw [->,blue] (9) to [out=235,in=340] (37); 
\draw [->,blue] (37) to [out=340,in=100] (2); 

\draw [->,magenta] (26) to [out=150,in=30] (25);      
\draw [->,magenta] (25) to [out=150,in=30] (24);     
\draw [->,magenta] (24) to [out=150,in=30] (23);    

\draw [->,magenta] (23) to [out=180,in=290] (38); 
\draw [->,magenta] (38) to [out=270,in=30] (22); 

\draw [->,magenta] (22) to [out=185,in=60]  (16);  
\draw [->,magenta] (16) to [out=235,in=130] (10); 
\draw [->,magenta] (10) to [out=235,in=130] (3); 

\draw [->,orange] (30) to [out=150,in=30] (29);      
\draw [->,orange] (29) to [out=150,in=30] (28);     
\draw [->,orange] (28) to [out=150,in=30] (27);    

\draw [->,orange] (27) to [out=150,in=60] (38); 
\draw [->,orange] (38) to [out=210,in=100] (22); 

\draw [->,orange] (22) to [out=235,in=130]  (17);  
\draw [->,orange] (17) to [out=235,in=130] (11); 
\draw [->,orange] (11) to [out=235,in=130] (4); 

\draw [->,teal] (33) to [out=150,in=30] (32);      
\draw [->,teal] (32) to [out=150,in=30] (31);

\draw [->,teal] (31) to [out=185,in=60] (27);    
\draw [->,teal] (27) to [out=235,in=340] (38); 
\draw [->,teal] (38) to [out=340,in=100] (23); 

\draw [->,teal] (23) to [out=235,in=130]  (18);  
\draw [->,teal] (18) to [out=235,in=130] (12); 
\draw [->,teal] (12) to [out=235,in=130] (5); 

\draw [->,brown] (35) to [out=150,in=30] (34);      
\draw [->,brown] (34) to [out=185,in=60](31);     
\draw [->,brown] (31) to [out=235,in=130] (28);    
\draw [->,brown] (28) to [out=235,in=130]  (24);   
\draw [->,brown] (24) to [out=235,in=130]  (19);  
\draw [->,brown] (19) to [out=235,in=130] (13); 
\draw [->,brown] (13) to [out=235,in=130] (6); 

\draw [->,violet] (36) to [out=185,in=60] (34);     
\draw [->,violet] (34) to [out=235,in=130](32);     
\draw [->,violet] (32) to [out=235,in=130] (29);    
\draw [->,violet] (29) to [out=235,in=130]  (25);   
\draw [->,violet] (25) to [out=235,in=130]  (20);  
\draw [->,violet] (20) to [out=235,in=130] (14); 
\draw [->,violet] (14) to [out=235,in=130] (7); 

\node [mylab,below of = 1, node distance=0.12in] (blank){\footnotesize{$t_1,t_2$} };
\node [mylab,below of = 2, node distance=0.12in] (blank){\footnotesize{$t_3$} };
\node [mylab,below of = 3, node distance=0.12in] (blank){\footnotesize{$t_4$} };
\node [mylab,below of = 4, node distance=0.12in] (blank){\footnotesize{$t_5$} };
\node [mylab,below of = 5, node distance=0.12in] (blank){\footnotesize{$t_6$} };
\node [mylab,below of = 6, node distance=0.12in] (blank){\footnotesize{$t_7$} };
\node [mylab,below of = 7, node distance=0.12in] (blank){\footnotesize{$t_8$} };

\node [mylab,right of = 36, node distance=0.12in] (blank){\footnotesize{$s_8$} };
\node [mylab,right of = 35, node distance=0.12in] (blank){\footnotesize{$s_7$} };
\node [mylab,right of = 33, node distance=0.12in] (blank){\footnotesize{$s_6$} };
\node [mylab,right of = 30, node distance=0.12in] (blank){\footnotesize{$s_5$} };
\node [mylab,right of = 26, node distance=0.12in] (blank){\footnotesize{$s_4$} };
\node [mylab,right of = 21, node distance=0.12in] (blank){\footnotesize{$s_3$} };
\node [mylab,right of = 15, node distance=0.12in] (blank){\footnotesize{$s_2$} };
\node [mylab,right of = 8, node distance=0.12in] (blank){\footnotesize{$s_1$} };

%
%
	
	\end{tikzpicture}

\vspace{-.3cm}	
	 \caption{An example of half-grid graph with two additional nodes.} \label{fig:half_grid_plus}
\end{figure}
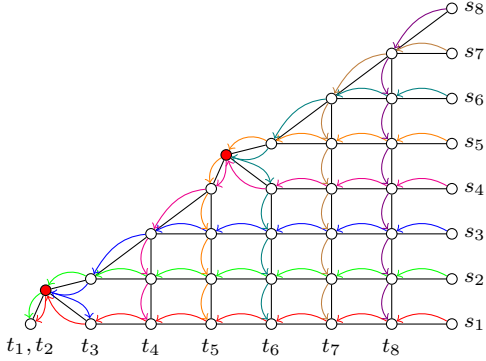

{\color{black}{
It must be noted that differently from  the case of arbitrary routing,
ICA is not always applicable to produce tight topologies, as additional requirements on the path length and number of paths are needed to ensure routing consistency.
Nevertheless, we can still use ICA for 
 certain values of $m$, $n$ and  $\bar{d}$, 
and obtain a  network topology that achieves the bound of Theorem \ref{th:bound_consistent_routing}. In particular we aim at creating a topology and routing scheme with the maximum number of nodes with unique encoding and minimum crossing number.

 First, we use ICA to generate the topology shown in Figure \ref{fig:half_grid}, that we name \emph{half-grid}. In this example, the number of paths is $m=8$. The figure highlights the source $s_i$ and destination $t_i$ of any path $p_i, \ i=1, \ldots, m$, where  $d_i =  8$ for all paths, hence $\bar{d}=m=8$. 
Observe that the half-grid satisfies the condition of routing consistency, and all the $n=\binom{8}{1}+\binom{8}{2}=36$ nodes are identifiable. 
In agreement with Observation \ref{obs:max_cross}, 
 the  maximum crossing number in this topology is  equal to $i_\texttt{max}=2$.

Such topology can be easily generalized by observing that its nodes are exactly those traversed by either one or two paths, hence it can be built for any $m$ paths, $n = m\cdot(m+1)/2$ nodes and $d_i=m$. In the resulting half-grid, routing is consistent and  all nodes are identifiable.

{\color{black}Then, in Figure \ref{fig:half_grid_plus}, we modified the half-grid of Figure \ref{fig:half_grid}, by adding two new nodes (the two red nodes of the figure)
using $m = 8$ paths, numbered as above, and  $d_1=\ldots=d_6=9$, $d_7 = d_8 = 8$, meaning that $\bar{d} = \frac{70}{8}=8.75$. 
Also in this case, we generated the node encodings in increasing order of the crossing number, and  the maximum crossing number is equal to $i_{\texttt{max}}=3$. Again, it holds that routing is consistent and that the bound of Theorem \ref{th:bound_consistent_routing} is achieved tightly, $\psiCPR={8 \choose 1} +{8 \choose 2} + \left\lfloor \frac{6}{3}\right\rfloor=38$.

We conclude that the topology of the half-grid can be modified by allowing paths to have longer lengths, adding some nodes with crossing number equal to 3 positioned in a way that 
routing is still consistent, while
the bound of Theorem \ref{th:bound_consistent_routing}  will still be tight.
 Notice that if $m\leq 4$, the half-grid topology meets the bound of Theorem \ref{th:bound_consistent_routing} for all values of $\bar{d}$. }

However, half-grid based topologies are not the only ones that can achieve the bound. 
An example is given in Figure \ref{fig:bacarozzost}  where ICA was used for  $m = 7$ consistent routing paths, each with length 12, except for one that has length 10, thus $\bar{d} = \frac{82}{7}$ and $d_\texttt{max}=12$. All the 39 nodes in the figure are identifiable, and so the bound of Theorem \ref{th:bound_consistent_routing} is achieved tightly and with nodes whose crossing number is always lower than or equal to 3. }
It remains open to find the general family of topologies that can achieve the bound in Theorem \ref{th:bound_consistent_routing}. 

\begin{figure}[t]
	\centering
	\includegraphics[width=0.75\columnwidth]{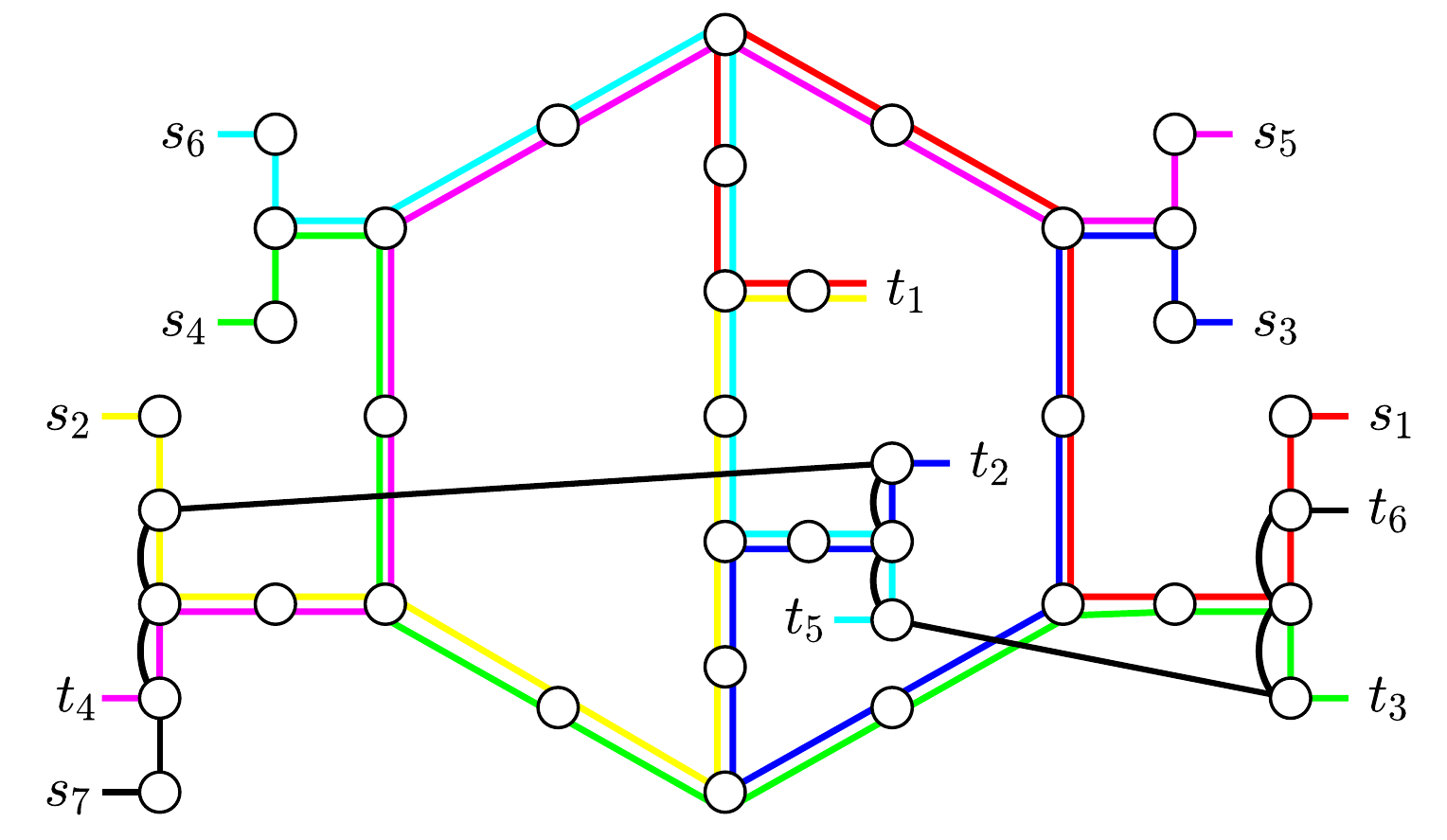}
\vspace{-.3cm}	
	 \caption{A topology that meets the bound of Theorem \ref{th:bound_consistent_routing} with $m=7$ and $\bar{d} = \frac{82}{7}$, and $d_\texttt{max}=12$.} \label{fig:bacarozzost}
\end{figure}

}}
{\color{black}
\subsection{Partially-consistent routing}

\label{subsec:part_cons_routing}

In this section, we relax the notion of routing consistency to provide a more general bound, which considers a limited number of violations of routing consistency.

\begin{definition}
\label{def:partially_consistent_routing}
If each path $p_i\in P$ can be divided into up to $q$ segments $s_1(p_i), s_2(p_i), \ldots, s_q(p_i)$, such that the property of routing consistency holds for the set $P_{\sfrac{1}{q}} = \cup_{p_i \in P} \{s_1(p_i), s_2(p_i), \ldots, s_q(p_i) \}$, then the routing scheme is called  $ \sfrac{1}{q}$~-~consistent.
\end{definition}

The following Lemma provides an analysis of the combinatorial patterns of consecutive ones under the assumption of $1/q$-consistent routing.

\begin{lemma}\label{le:kc1p}
Under the assumption of $\sfrac{1}{q}$-consistent routing, given a path $\hat{p}_i \in P$, all the columns $k=1, \ldots, m$ of the path matrix $M(\hat{p}_i)$ have up to $q$ sequences of consecutive ones.
\end{lemma}

\begin{proof}
Due to the $\sfrac{1}{q}$-consistency property of Definition \ref{def:partially_consistent_routing}, the  sub-matrices formed by the rows corresponding to the consistent routing segments of any path matrix will meet the consecutive ones property expressed by Lemma \ref{le:c1p}.  Therefore, $\sfrac{1}{q}$-consistency implies that each column can only have up to $q$ sequences of consecutive ones.
\end{proof}

In the following Lemma, we compute the maximum number of different encodings of a path matrix.

\begin{lemma}\label{le:1/n-numero}
Given a path $p_i \in P$ of length $d_i$, under the assumption of $\sfrac{1}{q}$-consistent routing,  with $m=|P| >1$ monitoring paths, the maximum number of  different encodings in the rows of $M(\hat{p}_i)$ is $\min \{2^{m-1}; 2q\cdot(m-1); d_i\}$. 
\end{lemma}
\begin{proof}

The number of different encodings in the rows of $M(\hat{p}_i)$ is  bounded by the length of $\hat{p}_i$, $d_i$.
As the $i$-th column of $M(\hat{p}_i)$ contains only ones, the different encodings in its rows can only be obtained by varying the values of the elements in the other columns. Accordingly, the number of different encodings in the rows of  $M(\hat{p}_i)$ is also bounded by $2^{m-1}$.
Furthermore, for any column $M(\hat{p}_i)|_{*,j}$ with $j \neq i$, it holds, by Lemma \ref{le:kc1p}, that it has at most $q$ sequences of consecutive ones. As a consequence, every column of $M(\hat{p}_i)$ can have no more than $2q$ flips. 
In order to have different encodings in any two successive rows $r$ and $r+1$ of the matrix $M(\hat{p}_i)$, that is $M(\hat{p}_i)|_{r,*} \neq M(\hat{p}_i)|_{r+1,*}$,
there must be at least a column  that flips in $r$.
Notice that the total number of columns that can flip is $m-1$ and each of them can flip no more than $2q$ times.
 When this bound is achieved, all columns other than the $i$-th column would have started from 0, flipped to 1, and then to 0 $q$ times.
The number of different  rows that can be observed in $M(\hat{p}_i)$  is therefore upper-bounded by $2q\cdot(m-1)$.
\end{proof}
We derive the upper-bound on the maximum number of identifiable nodes under partially-consistent routing in the following Theorem: 
\begin{theorem} [Partially-consistent routing]\label{th:bound_hcr}
In a general network with $n$ nodes,  $m>1$ monitoring paths and 
average path length $\bar{d}$, the number of identifiable nodes under $\sfrac{1}{q}$-consistent routing is upper bounded as in Theorem \ref{th:exp_max}, except that $N_{\texttt{max}}$ is replaced by 
$$N_{\texttt{max}} = 
m \cdot \min \{2^{m-1}; 2q \cdot (m-1); \bar{d}\}. $$
 \end{theorem}
\begin{proof}
The proof can be addressed as the one of Theorem \ref{th:exp_max}. 
 The maximum number of different encodings that can be observed in $m$ path matrices under the assumption of $\sfrac{1}{q}$-consistent path routing is bounded by
$N_{\texttt{max}} = 
m \cdot \min \{2^{m-1}; 2q \cdot (m-1); \bar{d}\} $, that is the sum for all paths of the bound shown in Lemma \ref{le:1/n-numero}. 
\end{proof}

In the particular case of $q=2$, we use the term {\em half-consistency}. Such a case is of particular interest. In fact, Al-Fares \emph{et al.} in \cite{Vahdat-fattree} proposed a half-consistent routing scheme to be adopted in fat-tree topologies, with the purpose to optimize bisection bandwidth. The proposed routing scheme spreads outgoing traffic among interconnected hosts as evenly as possible. We devote the following Section \ref{sec:fattree_experiments} to the analysis of half-consistent routing in fat-tree topologies.

Another motivating example for the study of  1/$q$-consistent routing is a multi-domain network with $q$ domains, in which routing consistency is guaranteed inside each domain, but inter-domain traffic can be split among multiple gateways between domains.
}

\subsection{A case study on half-consistent routing: fat-tree networks}

\color{black}{Typical data-center topologies are based on two or three levels of switches arranged into tree-like topologies. A common topology built of  commodity Ethernet switches is the {\em fat-tree} topology \cite{Leiserson}.}
\color{black}{Recent works on data-center design and optimization propose the use of fat-tree topologies to deliver high  bandwidth to  hosts  at the leaves of the fat-tree.
A special instance of a $k$-ary fat-tree together with a related addressing and routing scheme is described in the work of Al-Fares et al. in \cite{Vahdat-fattree}. Here the authors suggest the use of  homogeneous $k$-port switches to build the fat-tree topology and connect up to $k^3/4$ hosts.}
{\color{black} An example with 3 layers and  $k=4$ is shown  in Figure~\ref{fig:consistency_fattree}.}
%
{\color{black}In order to achieve maximum bisection bandwidth, which requires spreading the pod's outgoing traffic uniformly to the core switches, the authors of \cite{Vahdat-fattree} propose the use of a joint routing and addressing scheme which 
 violates the consistent routing assumption in two aspects:
(1) routes between different source-destination pairs may not be consistent, (2) routes in different directions between the same source-destination pair may not be consistent either. }

{\color{black}As an example consider the highlighted paths in Figure \ref{fig:consistency_fattree}. The blue path $p_1$ is used to send probing packets from the host $10.1.0.3$ to the host $10.3.1.3$.
$p_1$  consists of the following list of nodes:{\small{
\(p_1=<10.1.0.3, 10.1.0.1, 10.1.3.1, 10.4.2.1, 10.3.3.1, 10.3.1.1, 10.3.1.3>\)}}. Consider now, the red path $p_2$ that is used to send a packet from host $10.1.1.3$ to host $10.3.0.2$. This path consists of the following list of nodes:{\small{ $p_2= <
10.1.1.3, 10.1.1.1, 10.1.3.1, 10.4.2.2, 10.3.3.1, 10.3.0.1, 10.3.0.2>$}}.}
\begin{figure}[]
	\centering
	\includegraphics[width=\columnwidth]{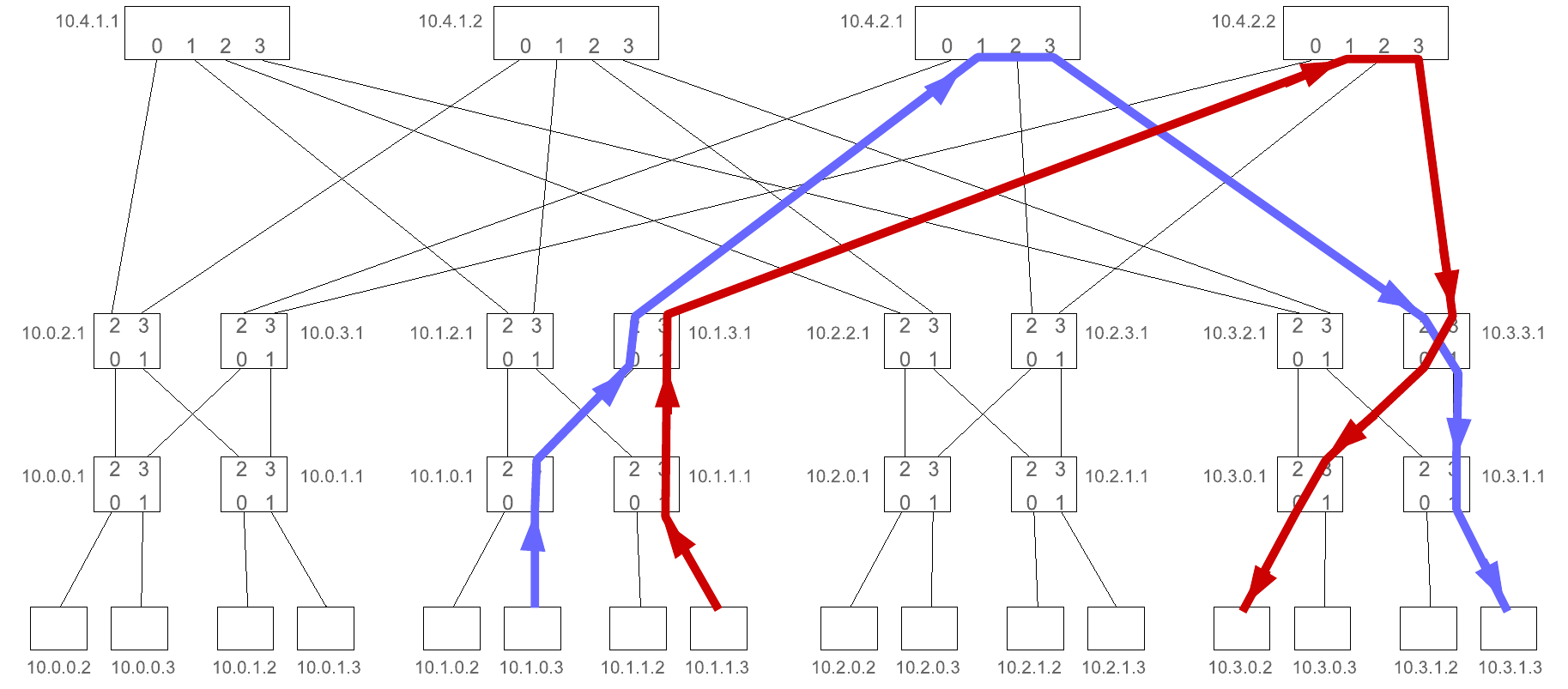} 
\vspace{-.7cm}
	\caption{Fat-tree with 3 layers and  $k=4$. Blue and red paths highlight routing inconsistency.} \label{fig:consistency_fattree}
	\vspace{-.5cm}	
\end{figure}
{\color{black}It follows that the routing scheme shown in Figure \ref{fig:consistency_fattree} is not consistent, as the path between the aggregation switches $10.1.3.1$ and $10.3.3.1$ can be different depending on the source and the destination hosts. Nevertheless, this is a case of half-consistent routing scheme, because the routing scheme only affects the choice of the core switches, while the other parts of the paths are fixed. }
\begin{prop}
Any shortest-path routing scheme on a fat-tree is half-consistent.
\label{prop:half_consistency}
\end{prop}
\begin{proof}
Let us call $u_s(p)$ and $u_t(p)$ the source and the destination endpoints of $p$, and let us call the {\em upper node} $u_m(p)$ the node of $p$ that is the farthest from the endpoints. Due to the structure of the fat-tree, there is only a unique path $s_1(p)$ from $u_s(p)$ to $u_m(p)$, and a unique  path $s_2(p)$ from $u_m(p)$ to $u_t(p)$. Therefore, for any two intermediate nodes on $s_i(p)$ ($i=1,\: 2$), there cannot be any alternative path between them, and the routing of these path segments is consistent.
\end{proof}

We devote Section \ref{sec:fattree_experiments}
to an experimental evaluation of identifiability bounds on fat-tree topologies.

\section{Service network monitoring}
\vspace{-.1cm}
We consider a service network where we monitor paths between clients and servers, under consistent routing in the case of (i) single-server and (ii) multi-server monitoring. 
{\color{black}To this purpose we refer to the work of  He et al. \cite{usICDCS16}, in which passive measurements along service paths are used to infer the status (working or not working) of the traversed nodes.
}

\subsection{Single-server monitoring }
\label{sec:client-server}
\subsubsection{Identifiability bound}

Consider the scenario where a single server communicates with multiple clients and we can only monitor the paths in between. The number of paths $m$ coincides with the number of clients, and all the monitoring paths must share a common endpoint (the server).

We start by showing the special structure of the topology spanned by the monitoring paths. 
\vspace{-.1cm}
\begin{lemma}\label{le:tree}
Under consistent routing, any monitoring paths with a common endpoint $r$ must form a tree rooted at $r$.
\end{lemma}
\vspace{-.4cm}
\begin{proof}
We consider any two paths $p_i$ and $p_j$. Starting from $r$, the next hops on these paths lead to either a common node or two different nodes. In the latter case, the two paths cannot intersect at any subsequent node $v$, as otherwise the two path segments from $r$ to $v$ following paths $p_i$ and $p_j$ would violate routing consistency. 
As this is true for all the paths, the paths must form a tree rooted at $r$.
\end{proof}
\vspace{-.1cm}
{\color{black}{
As a consequence many paths will have some common nodes and links, and this implies that the number of identifiable nodes with $m$ paths will be lower than in the general case expressed by Theorem
\ref{th:bound_consistent_routing}.
In the following (Theorem \ref{th:bound_with_tree}) we show that this number has indeed  an upper bound as small as  $2 m -1$.

Before we formalize this result let us introduce the concept of {\em optimal monitoring tree}, which is any tree topology (and related monitoring paths) that  guarantees the identifiability of all its nodes and for which the number of identifiable  nodes is maximum.
Given $m$ paths with maximum path length $d_{\texttt{max}}$, 
the optimal monitoring tree is a tree with $m$ leaves and maximum depth\footnote{The {\em depth of a tree} is the maximum  distance from the root to any leaf, in number of links.} $d_{\texttt{max}}-1$, that has the maximum number of identifiable nodes when its root-to-leaf paths are monitored.  }
}

\begin{lemma}\label{le:optimal_tree}
If the maximum path length $d_{\texttt{max}}$ satisfies $ d_{\texttt{max}} \geq \lceil \log_2 m\rceil +1$, the optimal monitoring tree is a full binary tree\footnote{We recall that a {\em full binary tree} is  a binary tree where each node is either a leaf or it has exactly two children.} with $m$ leaves.
If $d_{\texttt{max}} < \lceil \log_2 m \rceil +1$, then the optimal monitoring tree is a tree composed of $\left\lfloor \frac{m}{2^{(d_{\texttt{max}}-2)}}\right\rfloor$ perfect binary trees\footnote{We also recall that a {\em perfect binary tree} is a full binary tree where all leaves are at the same distance from the root.}  with depth $(d_{\texttt{max}}-2)$, and up to one full  binary tree with depth at most $(d_{\texttt{max}}-2)$ and $\left(m \mod 2^{(d_{\texttt{max}}-2)}\right)$ leaves, connected to  a common root.
\end{lemma}


\vspace{-.4cm}
\begin{proof}
Let us first consider the case of unbounded path length. By contradiction, assume the existence of an optimal monitoring tree that is not a full binary tree.
Such a tree must have at least a node $u$ whose number of children is either (a) strictly greater than two or it is (b) exactly one.

If (a), $u$ has at least three  children $v_1, v_2$ and $v_3$.
Let $p_1, p_2$ and $p_3$ be the paths from these nodes to $u$, as in Figure \ref{fig:l18_1}. We can build a new graph, starting from this,  with an additional identifiable node $x$, by removing the links between $u$ and $v_1$, $v_2$ and  adding $x$ as a parent of $v_1$ and $v_2$ and child of $u$. The modified topology is shown in  Figure \ref{fig:l18_2}.
Node $x$ is identifiable as its encoding is different from the encodings of the leaves $v_1$, $v_2$, as $x$ is traversed by the union of the set of paths traversing them, and from the encodings of $v_3$ and of the root $u$, as $x$ is not traversed by path $p_3$.
\begin{figure}[h!]
    \centering
    \begin{minipage}{.48\columnwidth}
        \centering
\begin{tikzpicture}[level/.style={sibling distance=1.2cm/#1}]
\tikzset{pallino/.style={draw,circle,fill=white,minimum size=4pt,inner sep=0pt}}
\tikzset{mylab/.style={black, minimum size=0.5cm}}

\node [pallino,circle,draw] [label=left:\footnotesize{$\LD 111 \RD$}] (u){}
  child {node [pallino,circle,draw][label=left:\footnotesize{$\LD 001 \RD$}] (v3) {}}
   child {node [pallino,circle,draw][label=left:\footnotesize{$\LD 010 \RD$}] (v2) {}}
    child {node [pallino,circle,draw][label=left:\footnotesize{$\LD 100 \RD$}] (v1) {}}; 
\node [mylab] at (-1.cm,0.5cm)   (8) {};
\node [mylab] at (1.,-0.7cm)   (6) {};

\draw [->,red] (u) to [out=210,in=80] (v3);      
\draw [->,green,dashdotted] (u) to [out=290,in=65] (v2);      
      
\draw [->,blue,dashed] (u) to [out=330,in=100] (v1);

 \node [mylab,below of = v3, node distance=0.13in] (blank){\footnotesize{$v_3$}};
\node [mylab,below of = v2, node distance=0.13in] (blank){\footnotesize{$v_2$} };
\node [mylab,below of = v1, node distance=0.13in] (blank){\footnotesize{$v_1$} };
\node [mylab,right of = u, node distance=0.13in] (blank){\footnotesize{$u$}};

\draw[blue,dashed] (-2.5cm, 0.5cm) -- (-1.8cm, 0.5cm) node[draw=none,fill=none] (pippo) [label=right:{{\color{black}$p_1$}}]{};
\draw[color=green,dashdotted] (-2.5cm, 0.2cm) -- (-1.8cm, 0.2cm) node[draw=none,fill=none] (pippo) [label=right:{{\color{black}$p_2$}}]{};
\draw[color=red] (-2.5cm, -0.1cm) -- (-1.8cm, -0.1cm) node[draw=none,fill=none] (pippo) [label=right:{{\color{black}$p_3$}}]{};
 
	\end{tikzpicture}  
  
        \vspace{-.3cm}   
        \caption{Three children tree}
        \label{fig:l18_1}
    \end{minipage}%
    \hfill
    \begin{minipage}{0.48\columnwidth}
        \centering
       
       \begin{tikzpicture}[level distance=1.0cm] 
  \tikzstyle{level 1}=[sibling distance=1.7cm] 
  \tikzstyle{level 2}=[sibling distance=1.2cm] 
  \tikzstyle{level 3}=[sibling distance=4mm] 
  \tikzset{pallino/.style={draw,circle,fill=white,minimum size=4pt,inner sep=0pt}}
\tikzset{mylab/.style={black, minimum size=0.5cm}}

\node [pallino,circle,draw] [label=left:\footnotesize{$\LD 111 \RD$}] (u){}
  child {node [pallino,circle,draw][label=left:\footnotesize{$\LD 001 \RD$}] (v3) {}}
   child {node [pallino,circle,draw][label=left:\footnotesize{$\LD 110 \RD$}] (x) {}
    child {node [pallino,circle,draw][label=left:\footnotesize{$\LD 010 \RD$}] (v2) {}}
    child {node [pallino,circle,draw][label=left:\footnotesize{$\LD 100 \RD$}] (v1) {}}};

   
   \draw [->,red] (u) to [out=210,in=80] (v3);     
  \draw [->,blue,dashed] (u) to [out=330,in=100] (x);
  
  \draw [->,blue,dashed] (x) to [out=330,in=100] (v1);      
 
   \draw [->,green,dashdotted] (x) to [out=210,in=80] (v2);      
   \draw [->,green,dashdotted] (u) to [out=290,in=160] (x);      
    
 
 \node [mylab,below of = v3, node distance=0.13in] (blank){\footnotesize{$v_3$}};
\node [mylab,below of = v2, node distance=0.13in] (blank){\footnotesize{$v_2$} };
\node [mylab,below of = v1, node distance=0.13in] (blank){\footnotesize{$v_1$} };
\node [mylab,right of = u, node distance=0.13in] (blank){\footnotesize{$u$}};
 \node [mylab,right of = x, node distance=0.13in] (blank){\footnotesize{$x$}};

	\end{tikzpicture}

        \vspace{-.3cm}
        \caption{Full binary tree}
        \label{fig:l18_2}
    \end{minipage}
\end{figure}
If (b), $u$ has only one child $v$, as shown in Figure \ref{fig:l18_bis_a}.
If $v$ is not traversed by any path, or all the paths traversing $u$ also traverse $v$, then node $v$ is not identifiable, and the removal of $v$ from the tree would not decrease the identifiability.
If instead there is a path $p_1$ traversing both $u$ and $v$, and a path $p_2$ traversing $u$ which ends before reaching node  $v$, as in Figure
\ref{fig:l18_bis_a}, 
 then path $p_2$ can be prolonged to traverse a new node $x$  added as a child of node $u$ to increase the identifiability of the topology, as shown in Figure \ref{fig:l18_bis_b}.

\begin{figure}[h!]
    \centering
    \begin{minipage}{.48\columnwidth}
        \centering
\begin{tikzpicture}[level/.style={sibling distance=60mm/#1}]
\tikzset{pallino/.style={draw,circle,fill=white,minimum size=4pt,inner sep=0pt}}
\tikzset{mylab/.style={black, minimum size=0.5cm}}

\node [pallino,circle,draw] [label=left:\footnotesize{$\LD 11 \RD$}] (z){}
  child {node [pallino,circle,draw][label=left:\footnotesize{$\LD 10 \RD$}] (a) {}};
   
\node [mylab] at (-1.cm,0.5cm)   (8) {};
\node [mylab] at (1.,-0.7cm)   (6) {};

\draw [->,red] (8) to [out=350,in=120] (z);      
\draw [->,red] (z) to [out=10,in=120] (6);
\draw [->,blue,dashed] (z) to [out=315,in=45] (a);      
  
  \node [mylab,above of = z, node distance=0.13in] (blank){\footnotesize{$u$}};
  
  \node [mylab,below of = a, node distance=0.13in] (blank){\footnotesize{$v$}};
                             
\draw[blue,dashed] (-2.5cm, 0.2cm) -- (-1.8cm, 0.2cm) node[draw=none,fill=none] (pippo) [label=right:{{\color{black}$p_1$}}]{};
\draw[color=red] (-2.5cm, -0.1cm) -- (-1.8cm, -0.1cm) node[draw=none,fill=none] (pippo) [label=right:{{\color{black}$p_2$}}]{};

	\end{tikzpicture}        
        \vspace{-.3cm}
        \caption{One child tree}
        \label{fig:l18_bis_a}
    \end{minipage}%
    \hfill
    \begin{minipage}{0.48\columnwidth}
        \centering
\begin{tikzpicture}[level/.style={sibling distance=15mm/#1}]
\tikzset{pallino/.style={draw,circle,fill=white,minimum size=4pt,inner sep=0pt}}
\tikzset{mylab/.style={black, minimum size=0.5cm}}

\node [pallino,circle,draw] [label=left:\footnotesize{$\LD 11 \RD$}] (z){}
  child {node [pallino,circle,draw][label=left:\footnotesize{$\LD 10 \RD$}] (a) {}}
   child {node [pallino,circle,draw][label=right:\footnotesize{$\LD 01 \RD$}] (b) {}}; 
\node [mylab] at (-1.cm,0.5cm)   (8) {};
\node [mylab] at (1.,-0.7cm)   (6) {};

\draw [->,red] (z) to [out=330,in=90] (b);      
   
\draw [->,blue,dashed] (z) to [out=210,in=90] (a);      
      
 \node [mylab,above of = z, node distance=0.13in] (blank){\footnotesize{$u$}};
  
  \node [mylab,below of = a, node distance=0.13in] (blank){\footnotesize{$v$}};  
    \node [mylab,below of = b, node distance=0.13in] (blank){\footnotesize{$x$}};        
                            

	\end{tikzpicture}            
            
              \vspace{-.3cm}
        \caption{Full binary tree}
        \label{fig:l18_bis_b}
    \end{minipage}
\end{figure}

Notice that as long as the maximum path length is
$d_{\texttt{max}} \geq \lceil  \log_2 m \rceil +1$, \color{black}{that is the unbounded case},
we can apply the previous discussion and  build an optimal full binary tree with up to $m$ leaves and depth $ \lceil  \log_2 m \rceil +1$ (maximum distance from the root to the leaves, in number of nodes). 
If instead $d_{\texttt{max}} < \lceil  \log_2 m \rceil +1$, the
largest number of leaves that can be obtained in a full binary tree topology with depth $d_{\texttt{max}}-1$ is $2^{d_{\texttt{max}}-1}$ which is lower than the number of paths $m$.
Therefore, in such a case,  the maximum identifiability is obtained by creating the maximum number $\lfloor \frac{m}{2^{(d_{\texttt{max}}-2)}}\rfloor$ of perfect binary trees of  depth $d_{\texttt{max}}-2$ and up to one full binary tree (not perfect) with depth at most $d_{\texttt{max}}-2$, connecting them to a same root, thus ensuring that the number of nodes with either no children or two only children is maximized.
\end{proof}

{\color{black}{\em Example: }
Figure \ref{fig:full_binary_tree_with_7_leaves}(a) shows an optimal monitoring tree for $m=7$ and $d_{\texttt{max}}=4$, i.e. a full binary tree.
In Figure \ref{fig:full_binary_tree_with_7_leaves}(b)  $m=7$ but $d_{\texttt{max}}=3$, so the optimal monitoring tree is made of $3$ perfect binary trees of depth 
$1$ and a full binary tree of depth at most $1$, connected to the same root.
}

\begin{figure}[h!]
\centering
\subfigure
{

\begin{tikzpicture}[level distance=0.5cm] 
  \tikzstyle{level 1}=[sibling distance=1.8cm] 
  \tikzstyle{level 2}=[sibling distance=0.9cm] 
  \tikzstyle{level 3}=[sibling distance=4mm] 
  \tikzset{pallino/.style={draw,circle,fill=white,minimum size=4pt,inner sep=0pt}}
\tikzset{mylab/.style={black, minimum size=0.5cm}}

\node [pallino,circle,draw] [label=above:\footnotesize{$s_1, \ldots, s_7$}] (u){}
child {node [pallino,circle,draw][] (u1) {}
child {node [pallino,circle,draw][] (u11) {}
child {node [pallino,circle,draw][] (u111) {}}
child {node [pallino,circle,draw][] (u112) {}}
}
child {node [pallino,circle,draw][] (u12) {}
child {node [pallino,circle,draw][] (u121) {}}
child {node [pallino,circle,draw][] (u122) {}}
}}
child {node [pallino,circle,draw][] (u2) {}
child {node [pallino,circle,draw][] (u21) {}
child {node [pallino,circle,draw][] (u211) {}}
child {node [pallino,circle,draw][] (u212) {}}
}
child {node [pallino,circle,draw][] (u22) {}}};
\node [mylab,below of = u111, node distance=0.13in] (blank){\footnotesize{$t_1$}};
\node [mylab,below of = u112, node distance=0.13in] (blank){\footnotesize{$t_2$}};
\node [mylab,below of = u121, node distance=0.13in] (blank){\footnotesize{$t_3$}};
\node [mylab,below of = u122, node distance=0.13in] (blank){\footnotesize{$t_4$}};
\node [mylab,below of = u211, node distance=0.13in] (blank){\footnotesize{$t_5$}};
\node [mylab,below of = u212, node distance=0.13in] (blank){\footnotesize{$t_6$}};
\node [mylab,below of = u22, node distance=0.13in] (blank){\footnotesize{$t_7$}};

	\end{tikzpicture}
}
\hspace{.5cm}
\subfigure{

\begin{tikzpicture}[level distance=0.7cm] 
  \tikzstyle{level 1}=[sibling distance=0.8cm] 
  \tikzstyle{level 2}=[sibling distance=0.4cm] 
  \tikzstyle{level 3}=[sibling distance=4mm] 
  \tikzset{pallino/.style={draw,circle,fill=white,minimum size=4pt,inner sep=0pt}}
\tikzset{mylab/.style={black, minimum size=0.5cm}}

\node [pallino,circle,draw] [label=above:\footnotesize{$s_1, \ldots, s_7$}] (u){}
child {node [pallino,circle,draw][] (u1) {}
child {node [pallino,circle,draw][] (u11) {}
}
child {node [pallino,circle,draw][] (u12) {}
}}
child {node [pallino,circle,draw][] (u2) {}
child {node [pallino,circle,draw][] (u21) {}
}
child {node [pallino,circle,draw][] (u22) {}
}}
child {node [pallino,circle,draw][] (u3) {}
child {node [pallino,circle,draw][] (u31) {}
}
child {node [pallino,circle,draw][] (u32) {}
}}
child {node [pallino,circle,draw][] (u4) {}
};

\node [mylab,below of = u11, node distance=0.13in] (blank){\footnotesize{$t_1$}};
\node [mylab,below of = u12, node distance=0.13in] (blank){\footnotesize{$t_2$}};
\node [mylab,below of = u21, node distance=0.13in] (blank){\footnotesize{$t_3$}};
\node [mylab,below of = u22, node distance=0.13in] (blank){\footnotesize{$t_4$}};
\node [mylab,below of = u31, node distance=0.13in] (blank){\footnotesize{$t_5$}};
\node [mylab,below of = u32, node distance=0.13in] (blank){\footnotesize{$t_6$}};
\node [mylab,below of = u4, node distance=0.13in] (blank){\footnotesize{$t_7$}};

	\end{tikzpicture}

 }
\center{\small{(a)} \hspace{3.6cm} \small{(b)}}
 \caption{Optimal monitoring tree:$m=7$ and $d_{\texttt{max}}=4$ (a) or $d_{\texttt{max}}=3$ (b).}
\label{fig:full_binary_tree_with_7_leaves}
\end{figure}
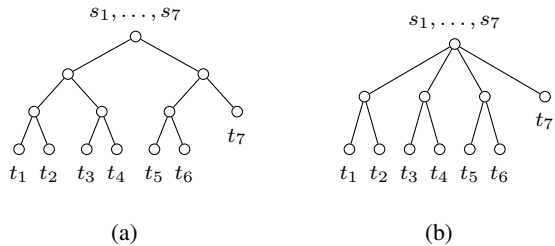 
{\color{black}
The following fact about full binary trees will be useful for bounding the identifiability in the case of single-server monitoring.

\begin{fact}\label{fact:num_fullbinary}
Given a full binary tree with $m$ leaves,
the number of nodes is $z_{\texttt{fb}}(m)\triangleq \max\{0,2m-1\}$.
\end{fact}
\begin{proof}
The fact can be proved by induction on $m$. 
If $m=1$ or $2$ the assertion is trivially true as the corresponding binary tree is unique and has $1$ or $3$ nodes, respectively.
Let us assume that the assertion is true for $m-1$ and $z_{\texttt{fb}}(m-1)=2(m-1)-1=2m-3$.
Let us now consider a generic full binary tree $t$ with $m$ leaves.
Let us remove the two leaves of any node $v$ of such a tree, obtaining the tree $t'$.
The tree $t'$ is also a full binary tree. $t'$ has  $m-1$ leaves as node $v$ is now a leaf itself.
Therefore the number of nodes of the new tree is $z_{\texttt{fb}}(n-1)=2m-3$. As the initial tree $t$ has two more nodes than $t'$ we can calculate  $z_{\texttt{fb}}(n)=2m-3+2=2m-1$ which concludes the proof. \qedhere

\end{proof}

Given the above properties we can formulate the following tight bound for the case of $m$ monitoring paths sharing a common endpoint, i.e, for single server monitoring.}

{\color{black}
\begin{theorem} [Identifiability for single-server monitoring]
\label{th:bound_with_tree}
Consider monitoring paths between a server and $m$ clients in a network of $n$ nodes and maximum path length $d_{\texttt{max}}$.  Then the maximum number of identifiable nodes $\psiTREE(m,n,d_{\texttt{max}})$ is upper-bounded by:
\begin{align}
 \left\{\begin{array}{l}
\min \left\{z_{\texttt{fb}}(m), n \right\}, \ \mbox{if }d_{\texttt{max}}\geq \lceil log_2 m \rceil+1\\
\min \left\{n; 1+ \Big\lfloor \frac{m}{2^{(d_{\texttt{max}}-2)}}\right\rfloor
\cdot z_{\texttt{fb}}(2^{(d_{\texttt{max}}-2)})+ \\ 
\hspace{4.5em} + z_{\texttt{fb}}( m \mod 2^{(d_{\texttt{max}}-2)})\Big\}, \  \mbox{otherwise}
\end{array}
 \right.
\end{align}
where $z_{\texttt{fb}}(m)\triangleq \max\{0,2m-1\}$ is the number of nodes in a full binary tree with $m$ leaves. 
\end{theorem}
{\color{black}
\begin{proof}

Let us first consider the case of unbounded path length.
Due to Lemma \ref{le:tree} the monitoring paths form a tree topology. Since we are interested in the case of maximum identifiability with a given number of paths, Lemma \ref{le:optimal_tree} provides the case of full binary tree to maximize the number of identifiable nodes given $m$ monitoring paths. It follows from Fact \ref{fact:num_fullbinary} that such a number is either $z_{\texttt{fb}}(m)$ or $n$ whichever is the lowest.

In the case of bounded path length, if $d_\texttt{max}\geq \lceil \log_2 m \rceil+1$, we apply  Lemma \ref{le:optimal_tree} to see that 
the path length limit has no implications on the value of the bound, which therefore would still be $z_{\texttt{fb}}(m)$ or $n$, whichever is the lowest.

If instead $d_\texttt{max}<\lceil log_2 m \rceil+1$, according to Lemma  \ref{le:optimal_tree} we know that the topology that guarantees maximum identifiability can be obtained by creating several full binary tree of depth $d_\texttt{max}-1$ and connecting them to a unique root.
The maximum number of leaves of a full binary tree of depth $d_\texttt{max}-1$ is $2^{(d_\texttt{max}-2)}$.
Therefore with $m$ paths we can create $\left\lfloor \frac{m}{2^{(d_\texttt{max}-2)}}\right\rfloor$ full binary trees of depth $d_\texttt{max}-1$, each guaranteeing identifiability of a root plus $
z_{\texttt{fb}}(2^{(d_\texttt{max}-2)})$ nodes (according to Fact \ref{fact:num_fullbinary}) and a full binary tree with depth lower than $d_\texttt{max}$ with the remaining $[m \mod 2^{(d_\texttt{max}-2)}]$ leaves, of depth lower than $d_\texttt{max}$ which will ensure the identifiability of other $z_{\texttt{fb}}(
m \mod 2^{(d_\texttt{max}-2)})$ nodes.

\end{proof}
}}

\subsubsection{Tightness of the bound and design insights}
Under the constraint that monitoring paths  have a common endpoint, for any given number of monitoring paths $m$, maximum path length $d_{\texttt{max}}$, and sufficiently large $n$, it is possible to design a network topology according to the  structure of an optimal monitoring tree, as described by  Lemma \ref{le:optimal_tree}, with a number of identifiable nodes equal to the bound in Theorem
\ref{th:bound_with_tree}. 

In particular, if $d_{\texttt{max}} \geq \lceil log_2 m \rceil+1$
 the topology would be a full binary tree as in the example of Figure \ref{fig:full_binary_tree_with_7_leaves}(a), while if $d_{\texttt{max}} < \lceil log_2 m \rceil+1$ the topology would be the composition of $\lfloor {m \over 2^{(d_{\texttt{max}}-2)}} \rfloor$ perfect binary trees of depth $d_{\texttt{max}}-2$, and a full binary tree of depth at most $d_{\texttt{max}}-2$, connected to a common root, as in the example of Figure \ref{fig:full_binary_tree_with_7_leaves}(b).

\vspace{-.1cm}
\subsection{Multi-server monitoring}
{\color{black}
We now consider the case in which monitoring is performed through the paths of an overlay service network, with a set of $S$ servers,
where each server $s$  ($s=1,\ldots,S$) has $m_s$ clients. We want to determine an upper bound on the number of identifiable nodes that can be obtained by varying the location of the servers in $S$ and related clients.

\subsubsection{Identifiability bound}

Since all the paths going from the $m_s$ clients to a deployed server $s$ will have the same destination, under the assumption of consistent routing they will form a  tree with $m_s$ leaves, as shown in Lemma \ref{le:tree}. Hence, we will have $S$ such trees of paths intersecting each other to increase identifiability.

We analyze two subcases: (i) \emph{fixed client assignment}, where the number of clients $m_s$ for each server is predetermined, and (ii) \emph{flexible client assignment}, where the total number of clients $\sum_{s=1}^S m_s$ is fixed but the distribution across servers can be designed.

Let us first consider  the paths of one monitoring tree at a time. The following lemma holds.}}}

\begin{lemma} \label{le:number_id_in_path}
Let us consider a tree of $m_s$ monitoring paths. The maximum number of identifiable nodes along any one of the $m_s$ paths is $m_s$.
\end{lemma}

\begin{proof}
By induction on $m$ and considering the tree structure we can see that in order for the root to be identifiable, its children must have diverting paths, and so forth for every new level in the tree. Given that the maximum number of diverting paths is bounded by $m$, then $m$ is the maximum number of identifiable nodes that can be found along a single monitoring path. More specifically this bound is met tightly when the tree is an unbalanced full binary tree.\qedhere
\end{proof}
Following a similar approach to the analysis we made for the proof of Theorem \ref{th:exp_max}, we want to give a value of an upper bound $N_{\texttt{max}}$ on the sum of the number
of different encodings in each path matrix.

\begin{lemma}\label{le:length_binary}
Given a 
monitoring tree with $m_s$ leaves $v_i$, $i=1, \ldots, m_s$. Let $\ell_k$ be the maximum number of identifiable nodes on the path from the leaf  $v_k$ to the root $r$ (calculated in number of traversed nodes). Let 
$L\triangleq\sum_{k=1}^m  \ell_k$.
The value of $L$ is bounded above as follows:
$$
L \leq \psi(m_s)\triangleq \frac{m_s^2+3m_s-2}{2}.
$$

\end{lemma}
\begin{proof}

By induction, when $m_s=1$, $L=1$ and the assertion is trivially true. When $m_s=2$ it is also true, and the sum of the paths of the  tree is $L=4$.
Assume that the assertion is valid for all trees with $m_s-1$ leaves, which means  that $\psi(m_s-1)=(m_s^2+m_s-4)/2$.
Let $t$ by any tree with $m_s-1$ leaves, and $L(t)$ be the value of $L$ for such a tree.
Let us consider the addition of a new path $p_{m_s}$ to the tree $t$, to obtain a new tree $t'$ with $m_s$ paths. 
According to Lemma \ref{le:number_id_in_path}, the 
maximum length of the new path $p_{m_s}$
in terms of identifiable nodes is $m_s$.
In order for all its nodes to be identifiable, it is necessary for the new path  to cross $m_s-1$ identifiable nodes of the tree $t$ going from the root $r$ to a leaf $v$ at distance $m_s-1$ (in number of nodes) from $r$.
Let $p_v$ be the monitoring path of $t$ running from $v$ to $r$.
We can use the new path $p_{m_s}$ of $t'$ to produce a maximum increase of identifiability by considering two new leaves $v'$ and $v''$ appended to $v$.
Of these two leaves, the leaf $v'$ can be identified by prolonging the path $p_v$, while the leaf $v''$ can be identified by the new path $p_{m_s}$ only. 
According to this construction, the value of $L(t')$ is obtained my adding $m_s+1$ to the value of $L(t)$, where $m_s$ nodes are due to the length of the new path $p_{m_s}$ and the term $+1$ is due to the increase in the length of the path $p_v$.
$$L(t')=L(t)+(m_s+1).$$
Considering the inductive hypothesis according to which $L(t) \leq (m_s^2+m_s-4)/2$ , 
we obtain the proof of the assertion: \\
\noindent
\begin{math}
L(t') \leq \psi(m_s)=m_s+1+\frac{m_s^2+m_s-4}{2}=\frac{m_s^2+3m_s-2}{2}
\end{math}.
\end{proof}

Let us denote with $m \leq \sum_{i=1}^S m_i$ the total number of clients, where the inequality derives from the fact the  the same clients may be interested in multiple services.
We consider the number of unique encodings appearing in each path matrix of a service $i$,  with $m_i$ clients, where  $i=1,2, \ldots, S$. If these encodings represent nodes that appear also in other paths, the same encodings will also appear in their respective path matrices.
Thanks to Lemma \ref{le:length_binary} we derive the following lemma.

\begin{lemma}\label{le:icdcs-nmax}
Let us consider $S$ services with $m_i$ clients each, where $i=1,2, \ldots, S$, and a total of $m \leq \sum_{i=1}^S m_i$ clients.
The sum of the maximum numbers of different binary encodings in each of the $m$ path matrices (including repetitions across different matrices) is
$$N_{\texttt{max}} \triangleq \sum_{i=1}^S \left[\frac{m_i^2+3m_i-2}{2} + 2m_i \cdot (m-m_i)\right].$$
\end{lemma}
\begin{proof}

In each path matrix related to the client-server path of 
a given service $i =1, \ldots, S$, there are $m_i$ columns related to the paths of the same service and other $m-m_i$ columns related to paths of the other services.  The sequence of bits of these latter columns  may flip twice, due to Lemma \ref{le:c1p}.

As each of these columns flip potentially creates a new encoding with respect to the
encodings that the columns related to the $m_i$ paths of the same service would generate alone, these column flips contribute additional $2(m-m_i)$ encodings to each path matrix.

This occurs across all the $m_i$ path matrices, where the number of the potentially different encodings related to the first columns (over all the $m_i$ matrices) is $\psi(m_i)$ as detailed in Lemma \ref{le:length_binary} and where the column flips of all the other $m-m_i$ columns will add $2(m-m_i)\cdot m_i$ encodings.

We conclude that the $m_i$ path matrices of the same service would generate 
$\frac{m_i^2+3m_i-2}{2} + 2m_i \cdot (m-m_i)$ potentially different encondings with possible repetitions in the different path matrices.
As this holds for the path matrices of all the services we can derive the formula for  $N_\texttt{max}$.
\end{proof}

\begin{theorem}[Multiple servers, fixed client assignment]
\label{th:bound_with_services}
Let us consider $S$ servers with $m_s$ clients each, where $s=1,2, \ldots, S$, and a total of $m \leq \sum_{s=1}^S m_s$ clients.
Let also  $n=|N|$ be the total number of nodes and $\bar{d}$ the  average path length.%
The maximum number of identifiable nodes $\psiFOREST(\mathbf{m},n,\bar{d})$ is upper bounded as in Theorem \ref{th:exp_max}, except that $N_{\texttt{max}}$ is replaced by\\

\begin{math}
N_{\texttt{max}}=\min\left\{m \cdot \bar{d}; \sum\limits_{s=1}^S \left[\frac{m_s^2+3m_s-2}{2} + 2m_s (m-m_s)\right]\right\}.
\end{math}
\end{theorem}

\begin{proof}
The proof is analogous to that of Theorem \ref{th:exp_max}, where $N_{\texttt{max}}$ is given by Lemma \ref{le:icdcs-nmax}.
\end{proof}

{\color{black}{

In a more general case, each client can be assigned to any of $S$ available servers. In this case, a valid bound on the identifiable nodes corresponding to the monitoring paths between clients and servers 
 should hold for all the possible assignments of the clients to the servers.
 In the following we denote any of these assignments as an $S$-dimensional integer  vector  $\mathbf{m}$, where each element  $m_s$ 
 gives the number of clients assigned to the server $s \in \mathcal{ S}$, and where it holds that $ m_s \geq 0$, $\forall s \in \mathcal{S}$ and $\sum_{s=1}^{S} m_s=m$. 
}}

{\color{black}{
The following theorem aims at characterizing the maximum identifiability that can be achieved by means of passive monitoring through service paths, in a multi-server scenario, when every client can be assigned to any server.}}

{\color{black}{

\begin{theorem} [Identifiability for multi-server monitoring with flexible client assignment]\label{th:undistinguished_clients_rev}
Consider monitoring the paths between $S$ servers and $m$ clients with arbitrary client-server assignment in a network with $n$ nodes, with  average path length $\bar{d}$. Then the maximum number of identifiable nodes $\psiFOREST(m,S,n,\bar{d})$ is upper-bounded as in Theorem \ref{th:exp_max},
 except that $N_{\texttt{max}}$ is specified by
\begin{math}
N_{\texttt{max}}=\min\left\{m \cdot \bar{d}; m^2(2-\frac{3}{2S}) + 3m/2-S\right\}.
\end{math}
\end{theorem}

}}
{\color{black}{
\begin{proof}

 Let $\mathcal{A}$ be the set of possible assignments of $m$ clients to $S$ servers:
  \noindent
  \begin{math}\mathcal{A}=\{\mathbf{m}\in \mathbb{N} | m_s \geq 0,\textrm{ and }\sum_{s=1}^{S}   m_s =m \} .
 \end{math}



The bound on the number of identifiable nodes in the case of $S$ servers and undistinguished clients can be formulated as in Theorems \ref{th:bound_consistent_routing} and \ref {th:bound_with_services}, where
{
{
\begin{math}
N_{\texttt{max}}=\min\left\{m\bar{d}; \max_{\mathbf{m}\in {\mathcal{A}}}  \sum_{s=1}^S \left[\frac{m_s^2+3m_s-2}{2} + 2m_s \cdot (m-m_s)\right]\right\}.
\end{math}
}}
In order to calculate $N_{\texttt{Max}}$ we address the optimization, in the integer variables $m_s$, of the objective function
$f(\mathbf{m})=\sum_{s=1}^S \left[(m_s^2+3m_s-2)/2 + 2m_s \cdot (m-m_s) \right]=
2m^2+3m/2-S
-3/2 \sum_{s=1}^S m_s^2$ (obtained by replacing  $\sum_{s=1}^S m_s$ with $m$ where possible), under the constraint that $\mathbf{m}\in\mathcal{A}$.
A relaxation of this problem 
leads to the following solution: $m_s=m/S$, $\forall s=1, \ldots, S$, and an objective value  of
$m^2(2-\frac{3}{2S}) + 3m/2-S$, which is an upper bound to
$f(\mathbf{m})$,
 from which we derive the assertion of the theorem.
\end{proof}

\vspace{-.1cm}
\subsubsection{Design insights}

In a setting in which the $m$ monitoring paths connect a given number of servers to their  clients, the maximum identifiability is obtained by letting the branches of several server-rooted optimal monitoring trees intersect with each other, while satisfying the consistent routing assumption and the constraint on the average path length $\bar{d}$. 


While in the case of fixed client assignment to servers, the number of leaves of each tree is predetermined, in the case of flexible client assignment, the proof of
Theorem \ref{th:undistinguished_clients_rev} suggests that the highest identifiability is obtained through a uniform assignment of clients to servers. In terms of topology design this implies that the maximally identifiable topology would require uniformly sized monitoring trees.

%

\section{Performance evaluation}\label{sec:Performance Evaluation}

To evaluate the tightness of the proposed upper bounds, we compare them with lower bounds obtained by known heuristics on synthetic and real network topologies. \color{black}{Since the bound in Theorem~\ref{th:exp_max} is achievable under arbitrary routing (see Section~\ref{subsubsec:Design - arbitrary routing}), but it is higher than the bound in Theorem \ref{th:bound_consistent_routing} when applied to consistent routing schemes, we 
show it once in Figure \ref{fig:A} and we omit it in the rest of the evaluation.} \looseness=-1
\textcolor{black}{In all the experiments, where not otherwise stated we have a uniform path length, therefore  $d_i = \bar{d} = d_{\texttt{max}}$, and vary the number of paths.
}

\subsection{Consistent routing}\label{sec:exp_cons_rout}
We analyze the tightness of the upper bound in Theorems~\ref{th:exp_max} and \ref{th:bound_consistent_routing}
under consistent routing.
\color{black}{In Figure \ref{fig:A} the upper bounds (UB) computed as in Theorems \ref{th:exp_max} and \ref{th:bound_consistent_routing} are shown together with a lower bound (LB) obtained by placing monitoring endpoints as in Section~\ref{subsubsec:Design - consistent routing}. 
%
%
We vary the number of paths while fixing the {\color{black} average path length} $ \bar{d} =d_{\texttt{max}}= 12$.}

\color{black}{Notice that the upper bounds given by Theorems \ref{th:exp_max} and \ref{th:bound_consistent_routing} for $d_\texttt{max}=12$, are the same for $m=2,\ 3$, that is when $\min \{ d_i; 2^{m-1}\} = \min \{ d_i; 2\cdot (m-1)\}$, and for $m \geq 7$, that is the threshold above which $\min \{ d_i; 2 \cdot (m-1)\}=d_i = 12$. This result highlights how consistent routing reduces the maximum number of identifiable nodes as far as $\bar{d}$ is not too small. \\
The figure also shows the identifiability  of the half-grid topology, (see Figures \ref{fig:half_grid} and \ref{fig:half_grid_plus}). Notice that, as we pointed out in Section \ref{subsubsec:Design - consistent routing}, the bound on the number of identifiable nodes under the assumption of consistent routing (Theorem \ref{th:bound_consistent_routing}) is tight on the half grid topology when $m$ satisfies $\frac{m^2 +3m - 6}{m}\geq d_i$ (that in this example is when $m\geq 10 $) and when $m\leq 4$. The green triangle in the figure represents the topology shown in Figure \ref{fig:bacarozzost}. }
%
%
%
%

In Figure \ref{fig:bound_no_knowledge} we show, for the same network, how the bound of Theorem~\ref{th:bound_consistent_routing} varies with the number of monitoring paths $m$ and the \color{black}{maximum path length $d_{\texttt{max}}$.}
For small values of $d_{\texttt{max}}$ the bound has an almost linear growth with $m$.
For larger values of $d_{\texttt{max}}$ the bound shows two regions: an initial  super-linear growth for small values of $m$, and a linear growth for large values of $m$.
The figure also shows that while the number of paths $m$ has a major impact on the number of
identifiable nodes, the length of the monitoring paths has a significant impact only when $d_{\texttt{max}}$ is small, and  diminishing impact otherwise.

\begin{figure}[]
    \centering
    \begin{minipage}{.48\columnwidth}
        \centering
        \includegraphics[height=1.08\linewidth, angle=-90]{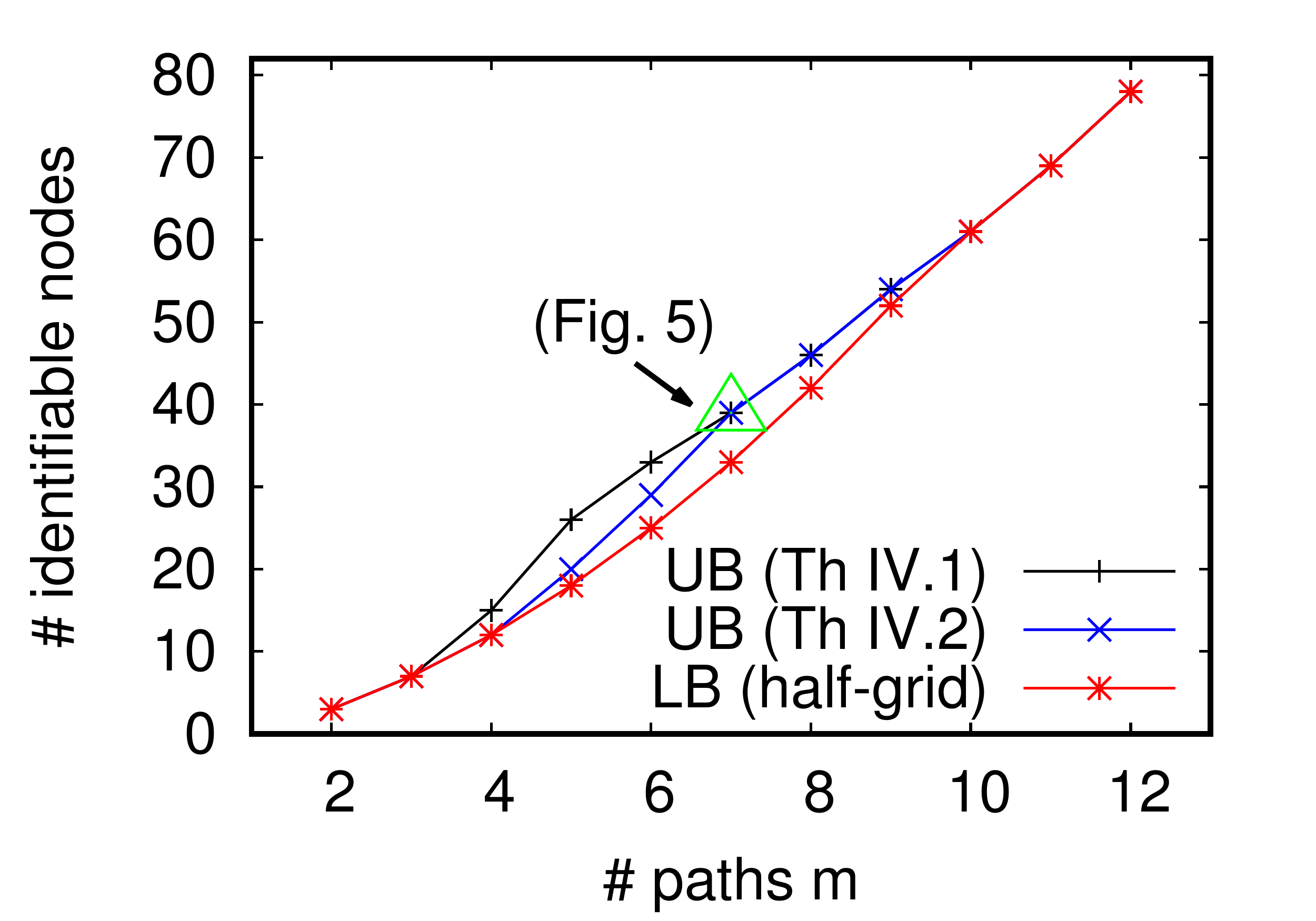}
        \caption{Bounds of Th. \ref{th:exp_max} and Th. \ref{th:bound_consistent_routing}, and LB for      
        $n=78$, varying $m$, and $d_{\texttt{max}}=12$.}
        \label{fig:A}
           \vspace{-.2cm}
    \end{minipage}%
     \hfill\hspace{.2cm}
    \begin{minipage}{0.48\columnwidth}
        \centering
        \includegraphics[height=1.08\linewidth, angle=-90]{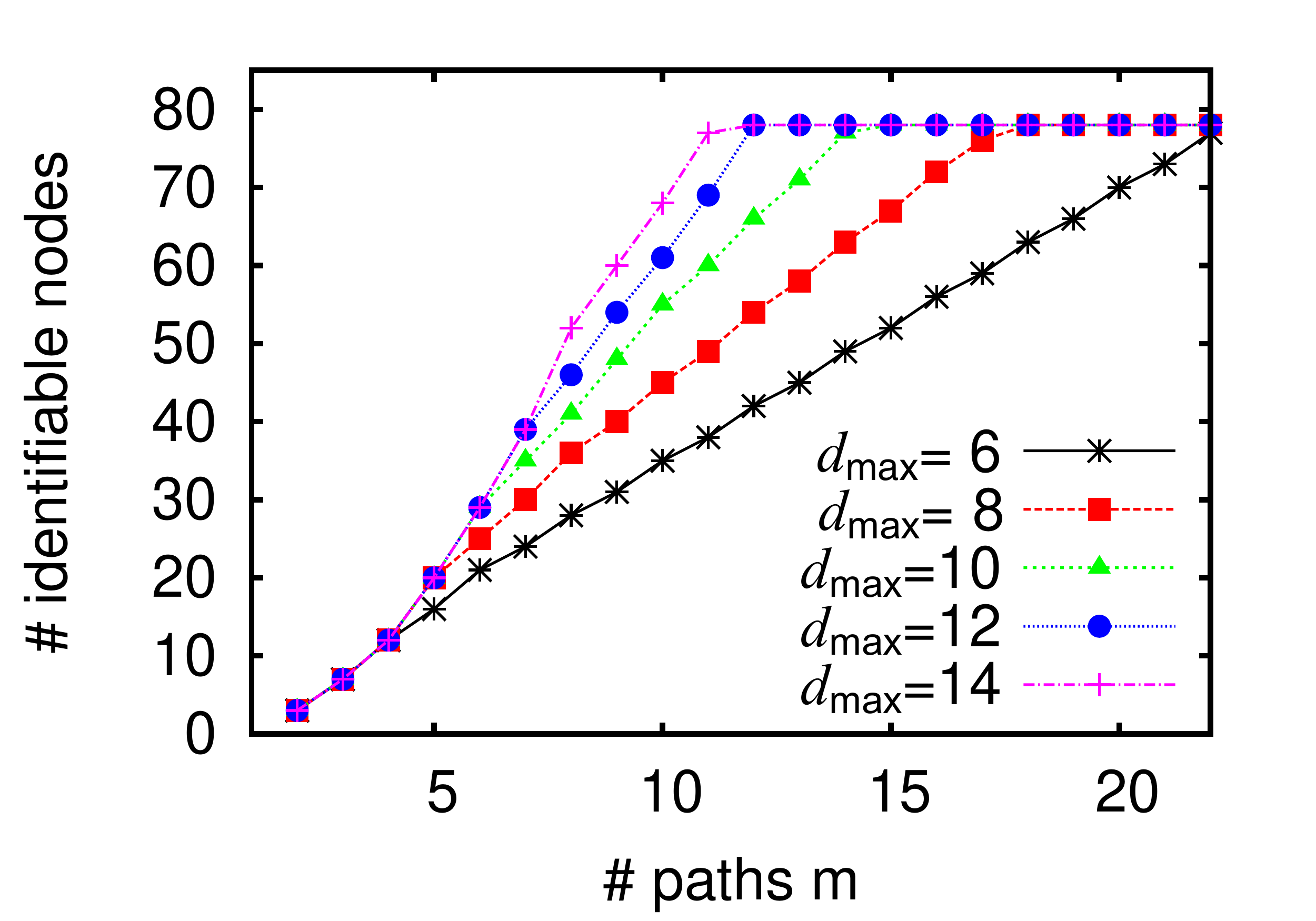}
        \caption{Bound of Th.~\ref{th:bound_consistent_routing}, for $n=78$, varying $m$ and different values of $d_{\texttt{max}}$.}
        \label{fig:bound_no_knowledge}
        \vspace{-.2cm}
    \end{minipage}
    \vspace{-.3cm}
\end{figure}

\subsection{Single-server monitoring} \label{sec:single_server_results}
{\color{black}
Figure \ref{fig:C} shows two scenarios with different topologies.
The first scenario is a {\color{black}network of 95 nodes}, connected as a full binary tree with 48 leaves, with $d_{\texttt{max}}= 7$ (in number of nodes). The figure shows the increase of the optimal number of identifiable nodes by varying the number of monitoring paths having a common endpoint.
By using 48 paths {\color{black}{each of length $d_i=7$}} from the leaves to the root, it is possible to identify all the  network nodes.
Notice that the optimal number of identifiable nodes that can be obtained by varying server location and placement of clients coincides with the  value of the bound of Theorem \ref{th:bound_with_tree}.
Lemma \ref{le:optimal_tree} shows in fact that for such a topology, the optimal identifiability is achieved by placing the endpoints of the $m$ different monitoring paths one in the root of the tree and the others in a way that the paths form a full binary tree topology. 

For the second scenario we consider a stricter limit on the path length: $d_i = d^*= d_{\texttt{max}} = 3$.
We consider a tree topology where a common root is connected to 24 binary trees of depth 1, for a total of 48 leaves, and 73 nodes (this topology is constructed extending the case of Figure \ref{fig:full_binary_tree_with_7_leaves}(b) to connect 24 subtrees).
In this topology, by using 48 paths {\color{black}{each of length $d_i = 3$}}, from the leaves to the root, it is possible to identify all the  nodes. 
{\color{black}Also in this case, the bound of Theorem \ref{th:bound_with_tree} is tight, and coincides with the optimal,
which is a tree of paths where $\lceil m/2 \rceil$  binary trees of depth 1 descend from a common root.
}
The Figure also shows that  the values of the bound obtained with Theorem \ref{th:bound_consistent_routing}, are considerably looser than those of Theorem \ref{th:bound_with_tree}. This is because the former considers any $m$ paths generated with any consistent routing scheme, while the latter considers the additional requirement that the monitoring paths share a unique common endpoint.


\begin{figure}
\subfigure
{
 \includegraphics[width=0.35\columnwidth, angle=-90]{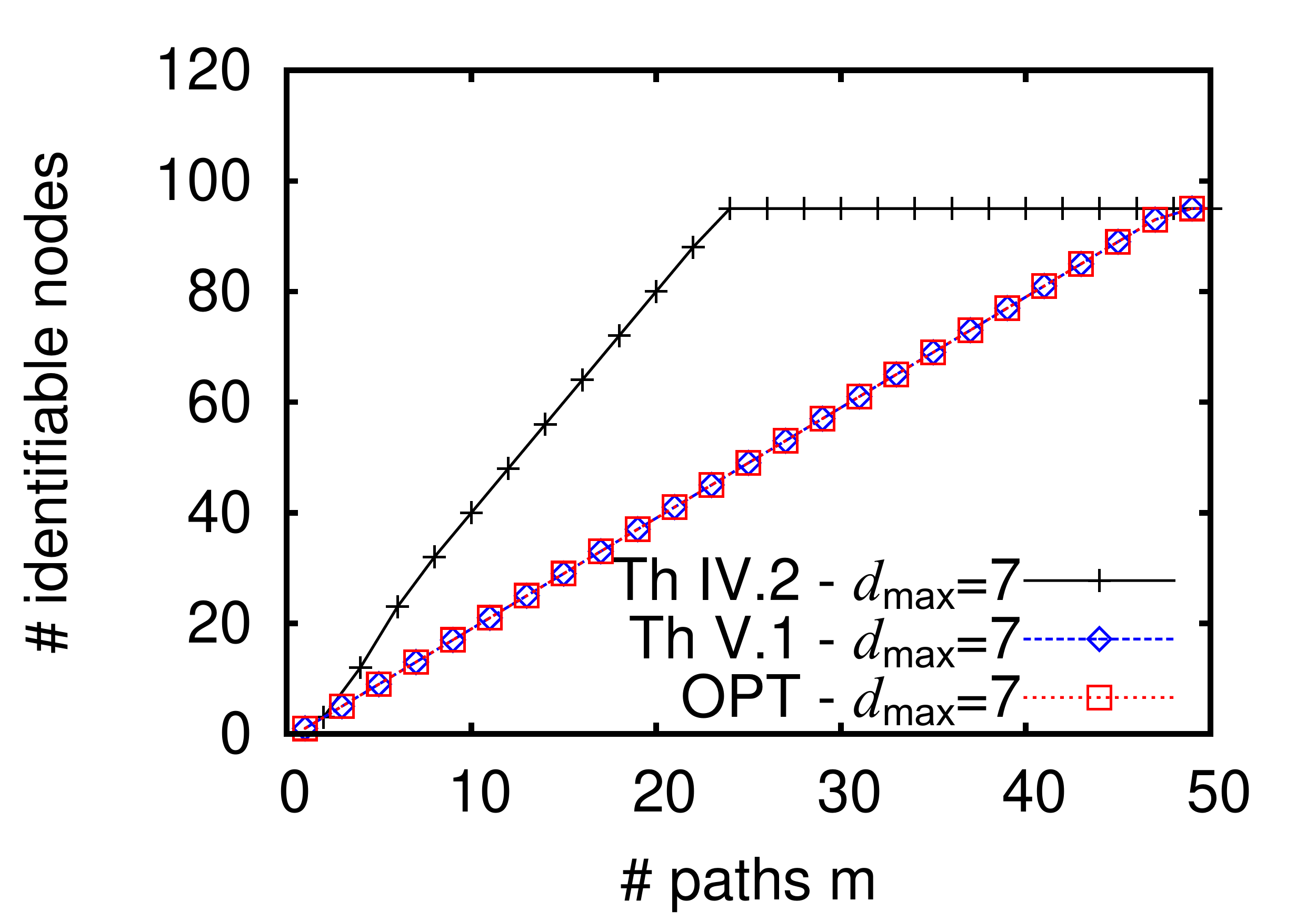}
}\hspace{-.58cm}
\subfigure{
\includegraphics[width=0.35\columnwidth, angle=-90]{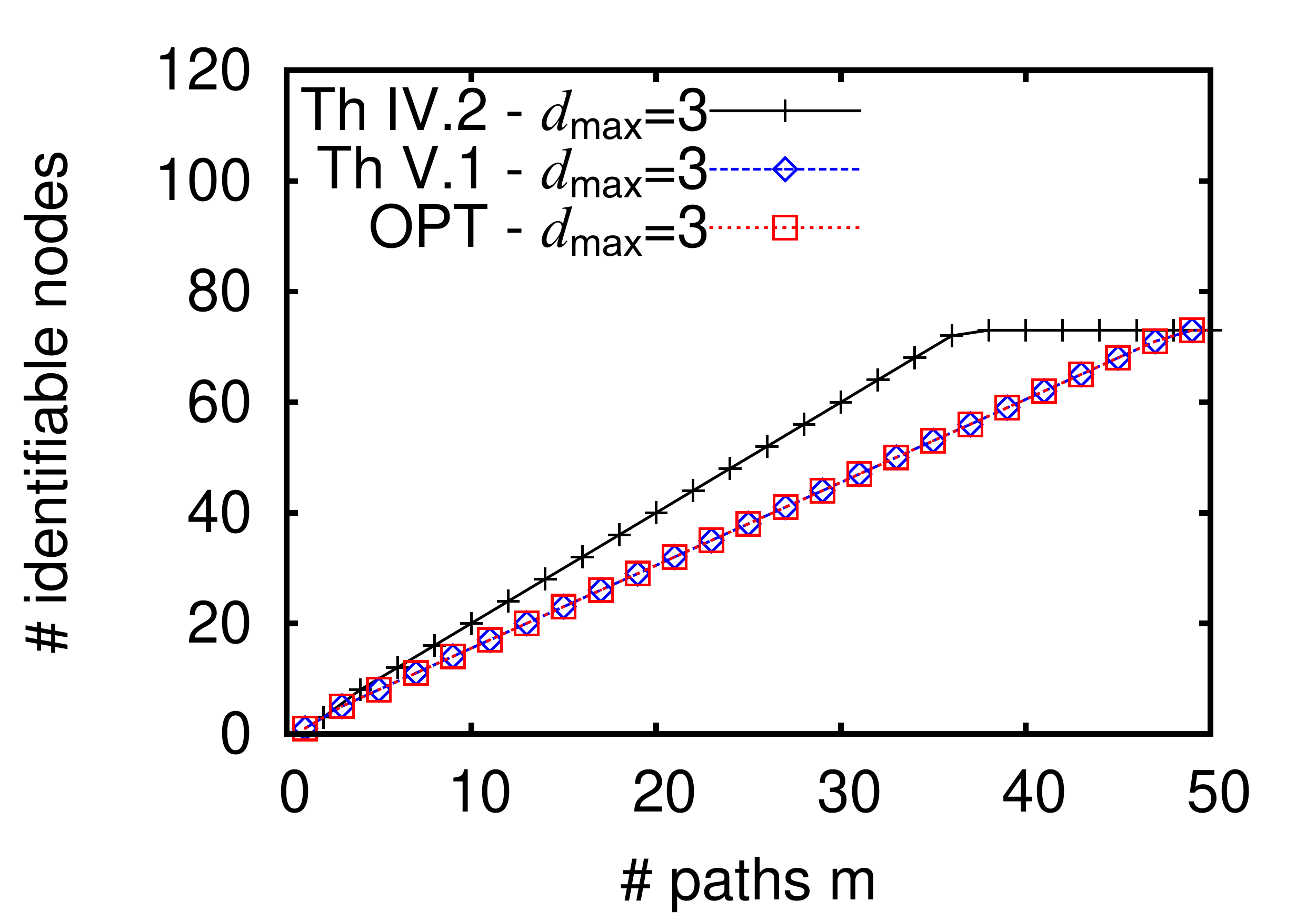}
 }
\center{\small{(a)} \hspace{4cm} \small{(b)}}
 \caption{Bound for single-server monitoring
(Th.~\ref{th:bound_with_tree}) - full binary tree for $d_{\texttt{max}}=7$ (a), multiple binary trees with a single root for $d_{\texttt{max}}=3$ (b).}
\label{fig:C} 
      \end{figure} 

%
%
%


\begin{figure}[tb]
    \centering
    \begin{minipage}{.48\columnwidth}
        \centering
        \includegraphics[height=1.08\linewidth, angle=-90]{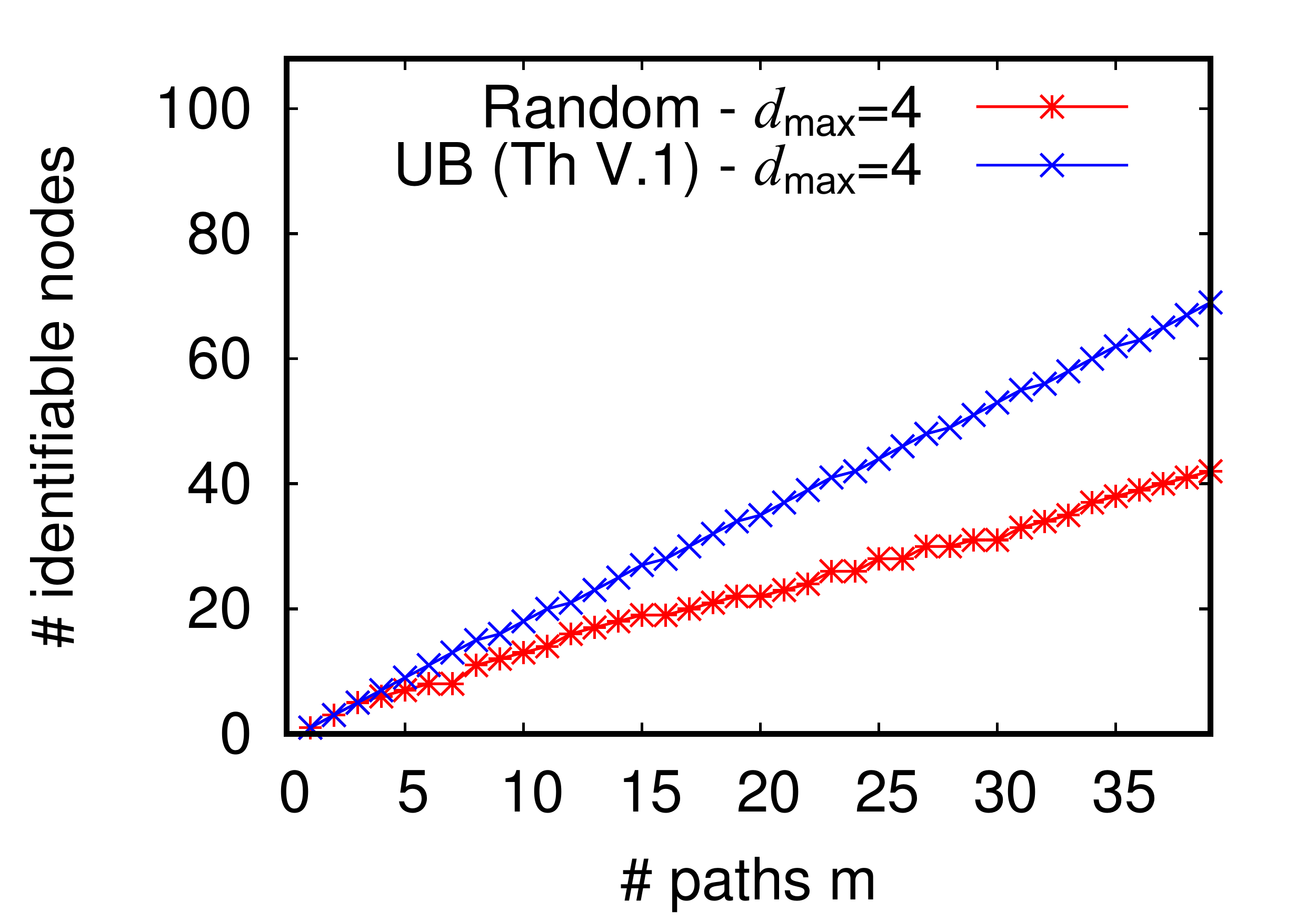}
        \caption{UB of Th~\ref{th:bound_with_tree} and LB of random placement, AT\&T  topology, $S=1$, $d_{\texttt{max}}=4$, varying $m$. }
        \label{fig:tree_monitoring}
    \end{minipage}%
    \hfill\hspace{.2cm}
    \begin{minipage}{0.48\columnwidth}
        \centering
        \includegraphics[height=1.08\linewidth, angle=-90]{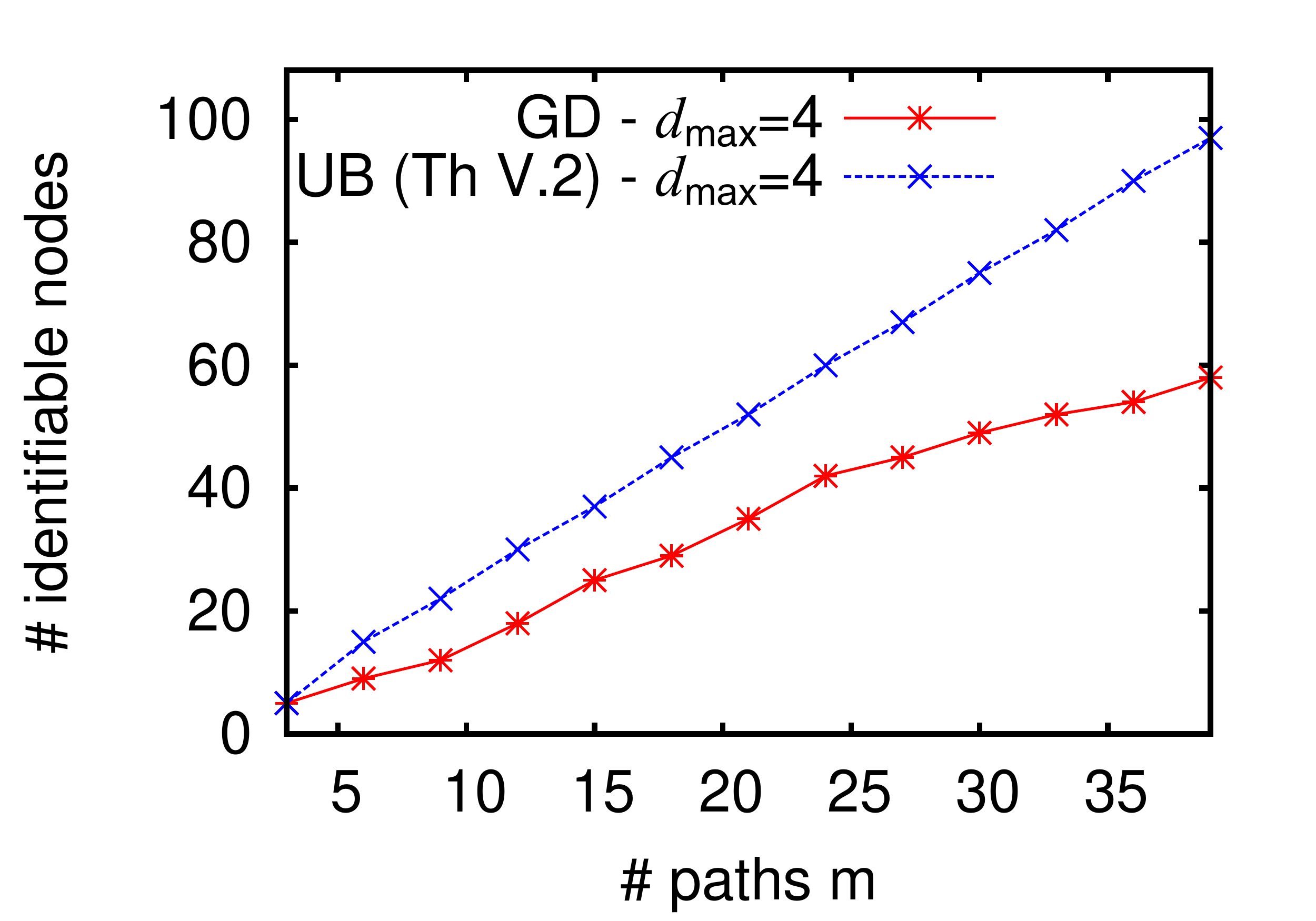}
        \caption{UB of Th. \ref{th:bound_with_services} and LB of GD \cite{usICDCS16}, AT\&T topology, $d_{\texttt{max}}=4$, varying $m$ and $S$ (3 clients per server). }
        \label{fig:bound_icdcs_m}
    \end{minipage}
\end{figure}

{\color{black}
Figure \ref{fig:tree_monitoring} illustrates
an experiment on
an existing AT\&T topology mapped with  Rocketfuel \cite{POP:02sigcom}, with 108 nodes and 141 links.
We consider a single server and a random placement of $m$ clients. We obtained a lower bound,
called "Random",
by running {\color{black}$100$} trials for each value of $m$ and using the largest number of nodes identified by client-server paths under consistent shortest path routing.  
We then compare this value to the upper bound given by Theorem~\ref{th:bound_with_tree}.
As the figure shows, the lower bound is not as close to the upper bound as in the case of the engineered topologies in Figure~\ref{fig:C}.
}
\subsection{Multi-server monitoring}
In these experiments we also consider
the AT\&T topology with 108 nodes and 141 links.
We  analyze the case of multiple servers, each serving  3 clients.
We increase the number of servers and vary the number of clients accordingly.
Figure \ref{fig:bound_icdcs_m} shows the upper bound of Theorem~\ref{th:bound_with_services}  compared to a lower bound obtained with the heuristic 
{\em greedy distinguishability maximization (GD)\footnote{Note that GD requires client locations to be predetermined. Here we place clients on some of the 78 dangling nodes, and then use GD to place servers. }}
proposed in \cite{usICDCS16}. Notice that this heuristic  finds a good approximation to the optimal number of identifiable nodes in this problem setting.
Although the heuristic only optimizes server placement, while Theorem~\ref{th:bound_with_services} considers the optimal placement of servers as well as clients, the experiment shows a good approximation of the upper and the lower bounds when $m$ is sufficiently small.

%
%

%

 {\color{black}
Figure
\ref{fig:alberi_uneven_distribution}
 shows a comparison of the three bounds of Theorems
\ref{th:bound_consistent_routing} (arbitrary sources/destinations),
 \ref{th:bound_with_services} (fixed client assignment) and \ref{th:undistinguished_clients_rev} (flexible client assignment) for the same topology, where
 we vary the numbers of services and clients, with an \color{black}{average} path length $\bar{d}=20$.
In the figure, the bound of Theorem \ref{th:bound_consistent_routing}
represents the special case of one client per server.
{\color{black}
We calculate the bound of Theorem~\ref{th:bound_with_services} assuming first a uniform assignment of clients to servers, as shown in Figure     
\ref{fig:alberi_uneven_distribution}(a), and then an uneven assignment, which is shown in Figure \ref{fig:alberi_uneven_distribution}(b).
For uneven assignment: in the case of two servers, one server is assigned to  4/5 of the clients, while the other to the rest 1/5; 
in the case of three servers, one server is assigned to 3/4 of the clients, the second server to 3/16, and the third server to 1/16.}
{\color{black}It can be seen that in the case of even assignment of clients to servers, the two bounds of Theorems
 \ref{th:bound_with_services} (fixed client assignment) and \ref{th:undistinguished_clients_rev} (flexible client assignment) give the same values. By contrast, in the case of
uneven distribution of clients to servers, Theorem~\ref{th:bound_with_services} gives a considerably smaller bound than Theorem~\ref{th:undistinguished_clients_rev}, which assumes an even distribution of clients to servers.



\begin{figure}[t]
\subfigure
{
 \includegraphics[width=0.35\columnwidth, angle=-90]{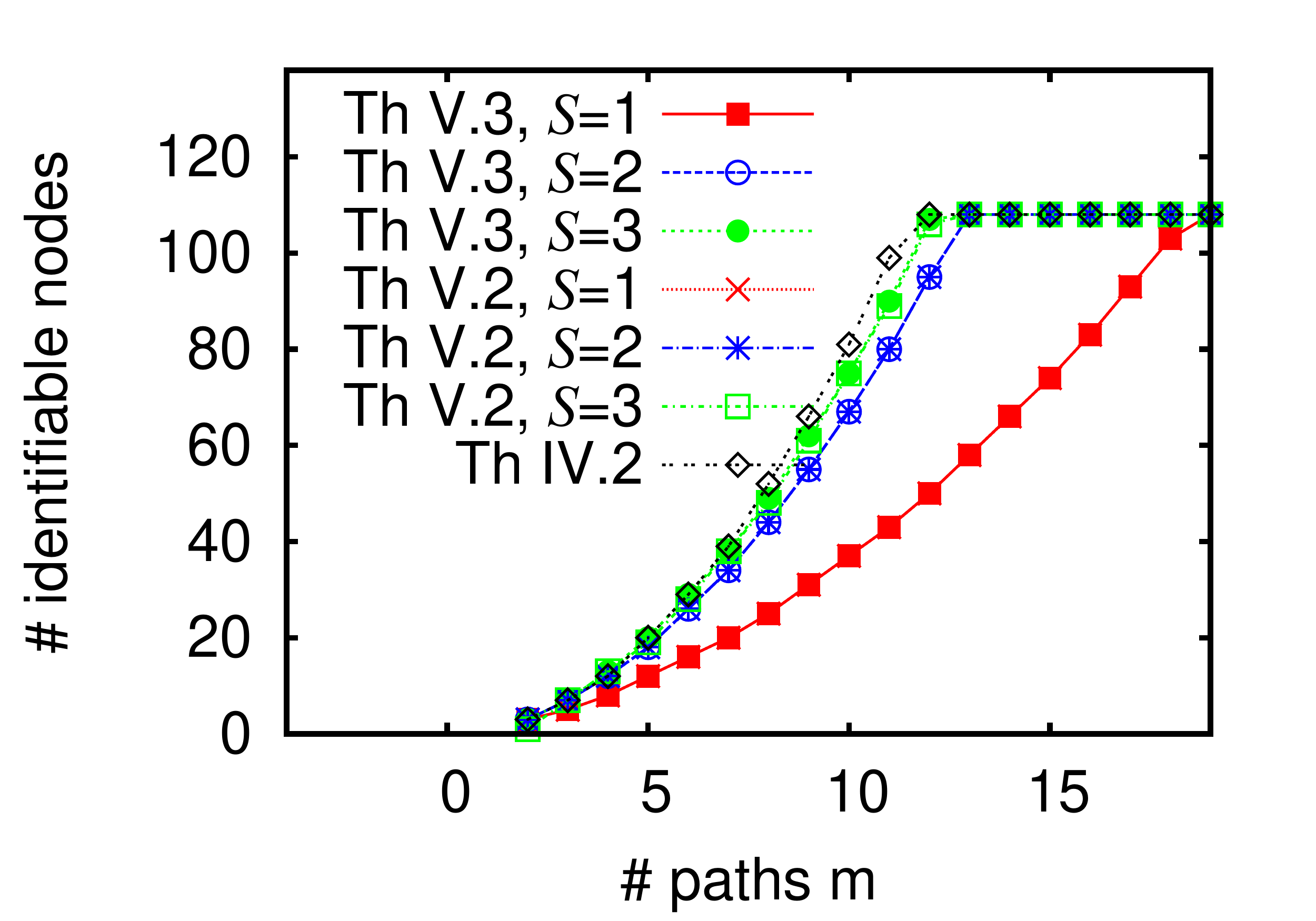}
}
\hspace{-.72cm}
\subfigure{
\includegraphics[width=0.35\columnwidth, angle=-90]{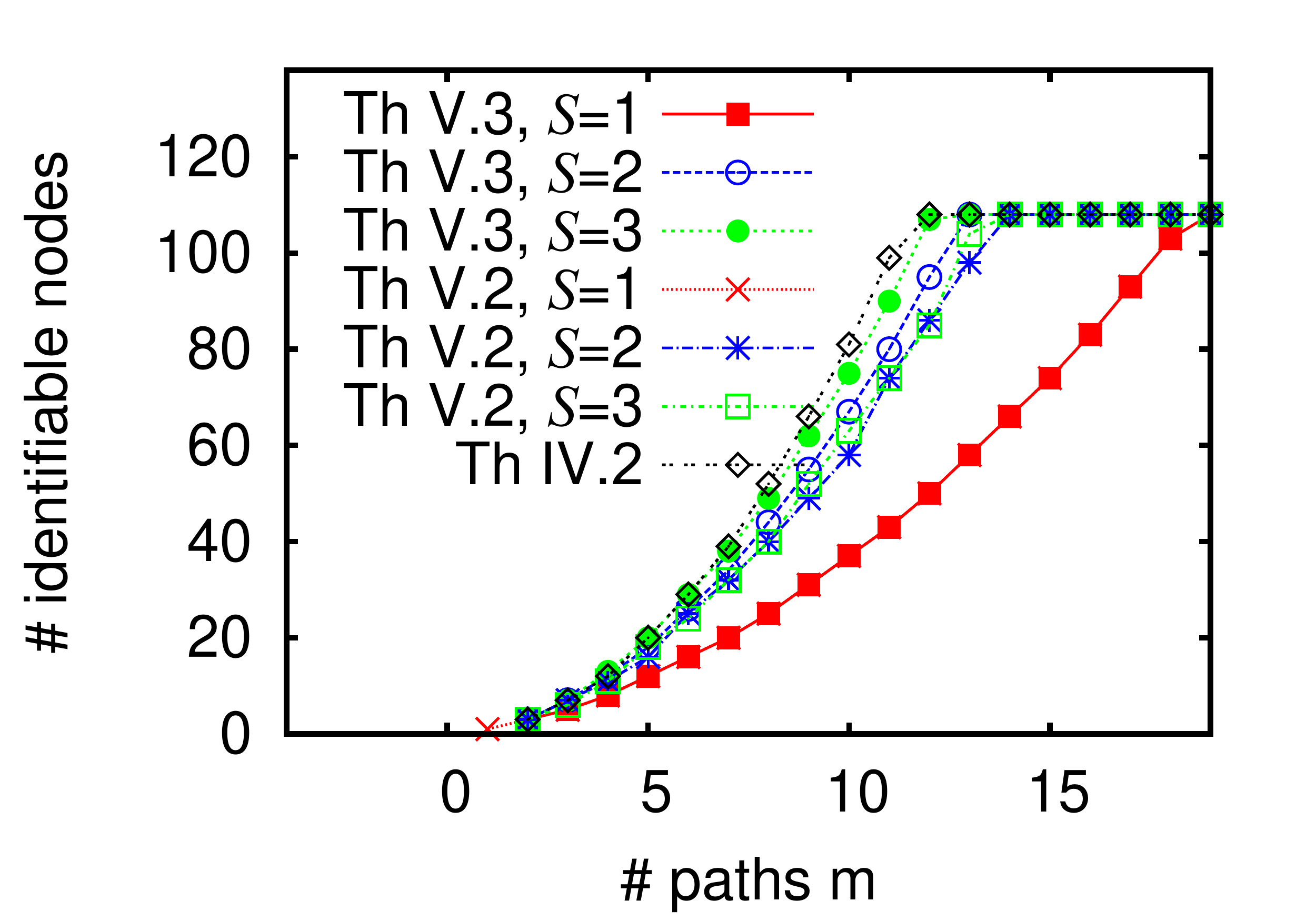}
 }
\center{\small{(a)} \hspace{4cm} \small{(b)}}
 \caption{UB of Theorems \ref{th:bound_consistent_routing}, \ref{th:bound_with_services}  and \ref{th:undistinguished_clients_rev}, AT\&T topology, $d_{\texttt{max}}$=20, $S$ servers, $m$ clients -  even (a) and uneven (b) distribution of clients to servers.}
\label{fig:alberi_uneven_distribution} 
      \end{figure}

%

\subsection{Data-center network monitoring}
\label{sec:fattree_experiments}

The identifiability of a fat-tree depends on the topology parameters $k$, $\ell$ and  the number of paths $m$.
In the following, we show that  only with a high number of layers, routing half-consistency plays a role in optimizing identifiability.
To this purpose Figure \ref{fig:consistency_d} evidences the difference in the upper bounds of the case of a more flexible half-consistent routing scheme considered in Theorem \ref{th:bound_hcr}, with respect to the case of
consistent routing considered in Theorem \ref{th:bound_consistent_routing}.
\begin{figure}[t]
    \centering
    \begin{minipage}{.48\columnwidth}
        \centering
        \includegraphics[height=1.08\linewidth, angle=-90]{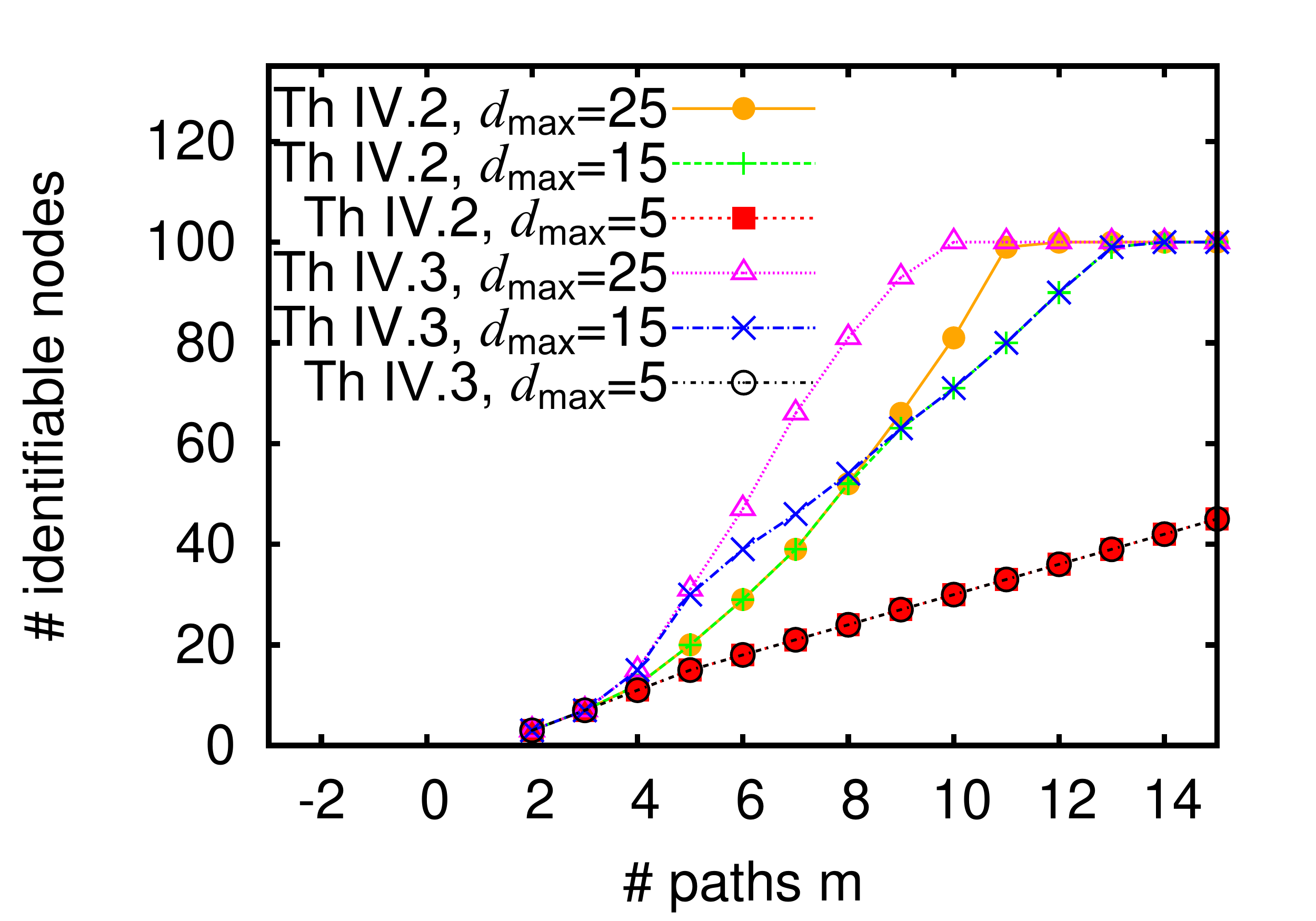}
        \caption{UB of Theorems \ref{th:bound_consistent_routing} 
        and \ref{th:bound_hcr} 
         - $100$ nodes, varying  $d_\texttt{max}$. }
        \label{fig:consistency_d}
        \vspace{-.3cm}
    \end{minipage}%
        \hfill\hspace{.2cm}
    \begin{minipage}{0.48\columnwidth}
        \centering
        \includegraphics[height=1.08\linewidth, angle=-90]{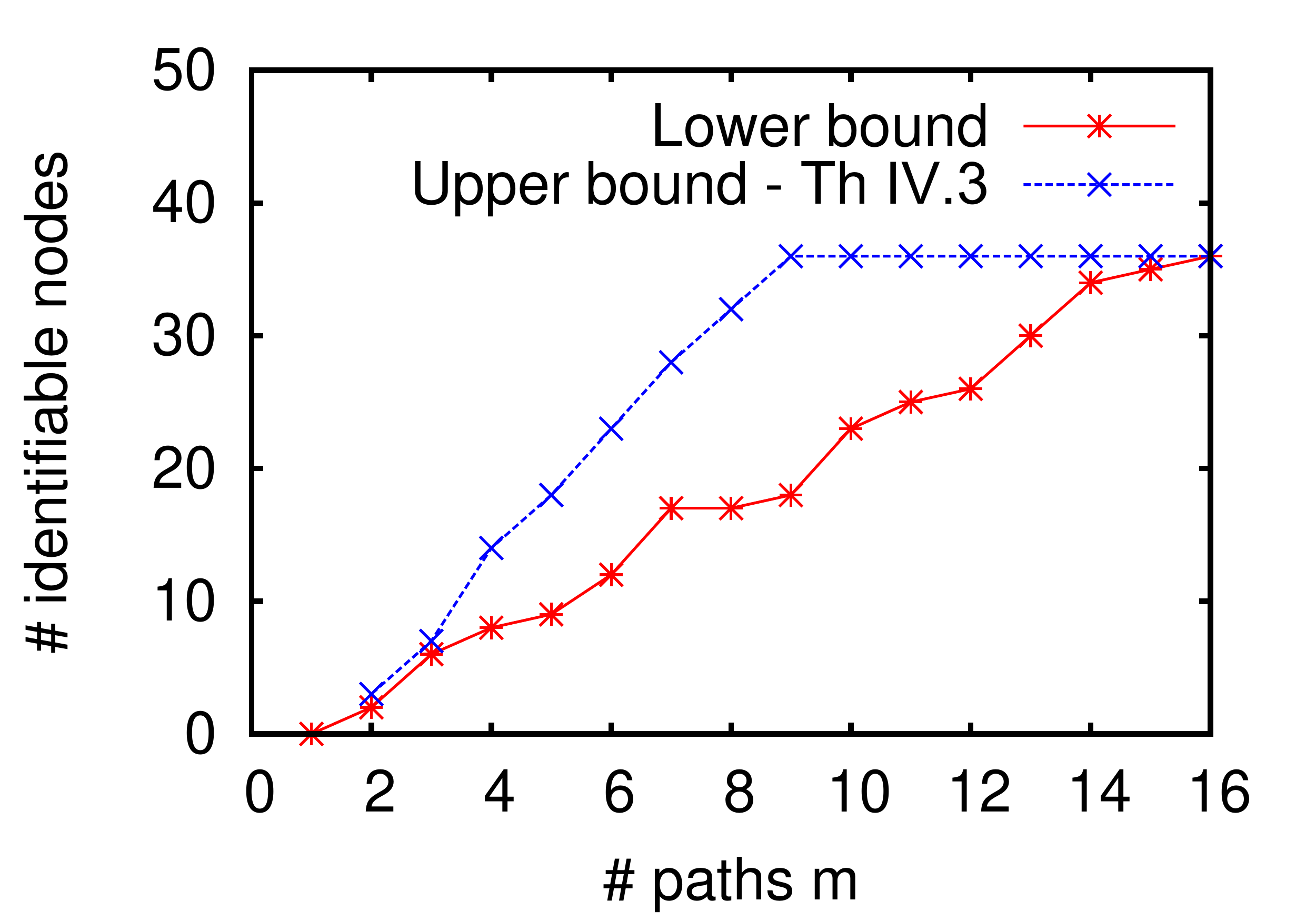}
        \caption{UB of Th.~\ref{th:bound_hcr} 
        and LB for a $4$-ary fat-tree with 3 layers.}
        \label{fig:comparison}
    \vspace{-.3cm}
    \end{minipage}
       \vspace{-.2cm}
\end{figure}
It considers a general network with 100 nodes. The difference of identifiability  between consistent and half-consistent routing grows by increasing  the maximum length of monitoring paths as {\color{black}{$d_{\texttt{max}} =5, 15, 25$}}, which in a fat-tree would correspond to values of $\ell=2, 7, 12$.
{\color{black}In conclusion, we can affirm that for topologies with very short diameter, such as in the case of fat-trees, 
having a higher degree of freedom in routing (half-consistent routing) has a significant impact on the identifiability of the network only for a high number of layers.}

%

We now consider the case in which monitoring is performed along paths  between hosts of a data-center network with a fat-tree topology and the routing scheme proposed in \cite{Vahdat-fattree}.
In Figure \ref{fig:comparison} we
consider a $4$-ary fat-tree with three layers
and study the tightness of the bound of Theorem \ref{th:bound_hcr}.
Due to the high complexity in selecting the optimal monitoring paths, we resort to {\color{black} an empirical selection of  paths that give us a lower bound on the  number of identifiable nodes}. 
It is interesting to see that  with only 16 monitoring paths we are able to monitor all the $36$ nodes of this fat-tree.

\section{Conclusion}\label{sec:Conclusion}
We consider the problem of  maximizing the number of nodes whose states can be identified  via Boolean network tomography.
We formulate the problem in terms of graph-based group testing and exploit the combinatorial structure of the testing matrix to derive upper bounds on the number of identifiable nodes under different assumptions, including: arbitrary routing, consistent routing, monitoring through client-server paths with one or multiple servers (and even or uneven distribution of clients), and half-consistent routing.
These  bounds show the fundamental limits of Boolean network tomography in both real and engineered networks.
We use the bound analysis to derive  
insights for the design of topologies with high identifiability in different network scenarios.
Through analysis and experiments we  evaluate the tightness of the bounds and demonstrate the  efficacy of the design insights for engineered as well as real networks.

\bibliographystyle{IEEEtran}
\bibliography{IEEEabrv,mybib}

\setlength{\intextsep}{0pt}
\vspace{-1.2cm}
\begin{IEEEbiography}[{\includegraphics[width=1in,height=1.25in,
clip,keepaspectratio]{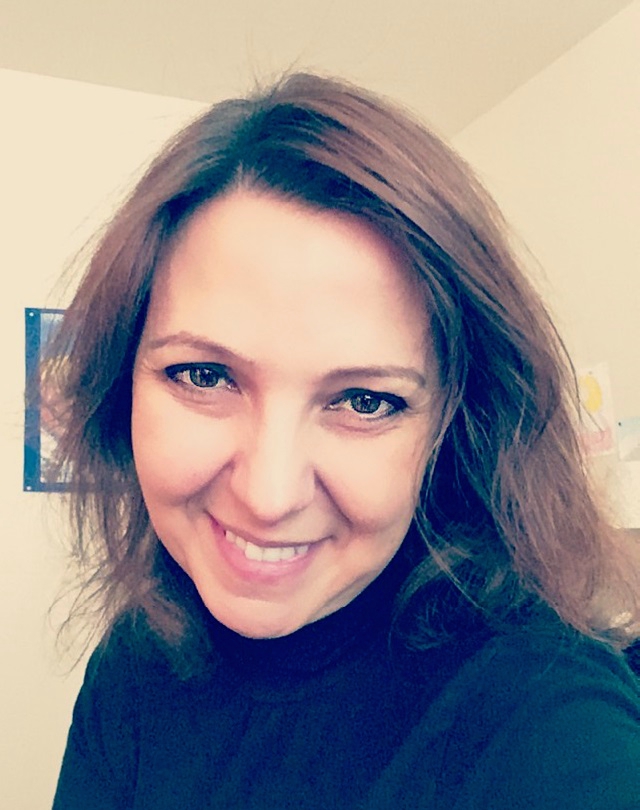}}]{Novella Bartolini (SM '16)}
graduated with honors in 1997 and received her PhD in
computer engineering in 2001 from the University of Rome, Italy. She is now
associate professor at Sapienza University of Rome. She was visiting professor at Penn State University for three years from 2014 to 2017.
Previously, she was  visiting scholar at the
 University of Texas at Dallas for one year in 2000 and research assistant at the University of Rome 
'Tor Vergata' in 2001-2002. She was program chair and program committee member of several international 
conferences. She
has served on the editorial board of Elsevier Computer Networks and ACM/Springer
Wireless Networks.
Her research interests lie in the area of wireless networks and network management. 
\end{IEEEbiography}

\vspace{-1.cm}
\begin{IEEEbiography}[{\includegraphics[width=1in,height=1.25in,
clip,keepaspectratio]{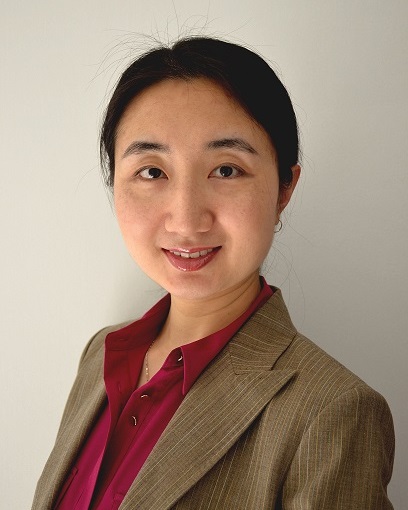}}]{Ting He (SM '13)} received the B.S. degree in computer science from Peking University, China, in 2003 and the Ph.D. degree in electrical and computer engineering from Cornell University, Ithaca, NY, in 2007.
Ting is an Associate Professor in the School of Electrical Engineering and Computer Science at Pennsylvania State University, University Park, PA. Between 2007 and 2016, she was a Research Staff Member in the Network Analytics Research Group at the IBM T.J. Watson Research Center, Yorktown Heights, NY. Her work is in the broad areas of network modeling and optimization, statistical inference, and information theory.
\end{IEEEbiography}

\vspace{-1.cm}
\begin{IEEEbiography}[{
\includegraphics[width=1in,height=1.25in,
clip,keepaspectratio]{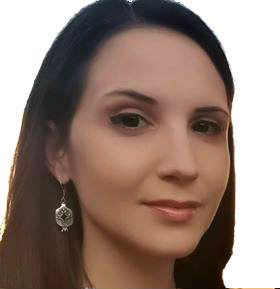}}]{Viviana Arrigoni}
 received the B.Sc. degree in Mathematics and the M.Sc. degree in
Computer Science from Sapienza, University of Rome, Italy. She is a Ph.D. student at the Department of Computer Science of the same university. Her research interests comprise computational Linear Algebra, Network topologies and Information Theory. 
\end{IEEEbiography}

\vspace{-1.cm}
\begin{IEEEbiography}[{\includegraphics[width=1in,height=1.25in,
clip,keepaspectratio]{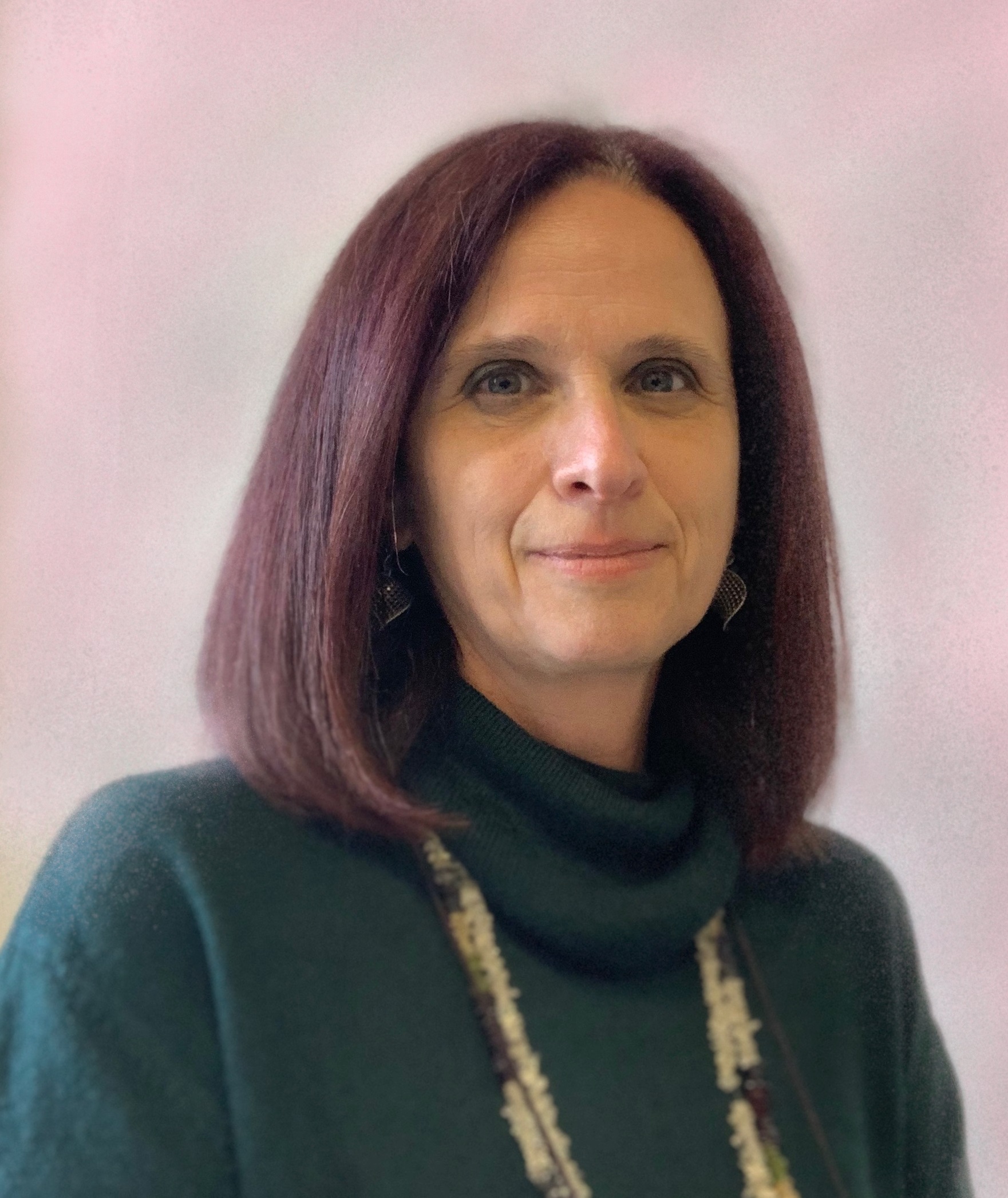}}]{Annalisa Massini}
received the degree in Mathematics and her Ph.D. in
Computer Science at Sapienza University of Rome, Italy, in 1989 and 1993 respectively.
Since 2001 she is associate professor at the Department of Computer Science of
Sapienza University of Rome.
Her research interests include hybrid systems, sensor
networks, networks topologies.
\end{IEEEbiography}
\vspace{-1.0cm}
\begin{IEEEbiography}[{\includegraphics[width=1in,height=1.25in,
clip,keepaspectratio]{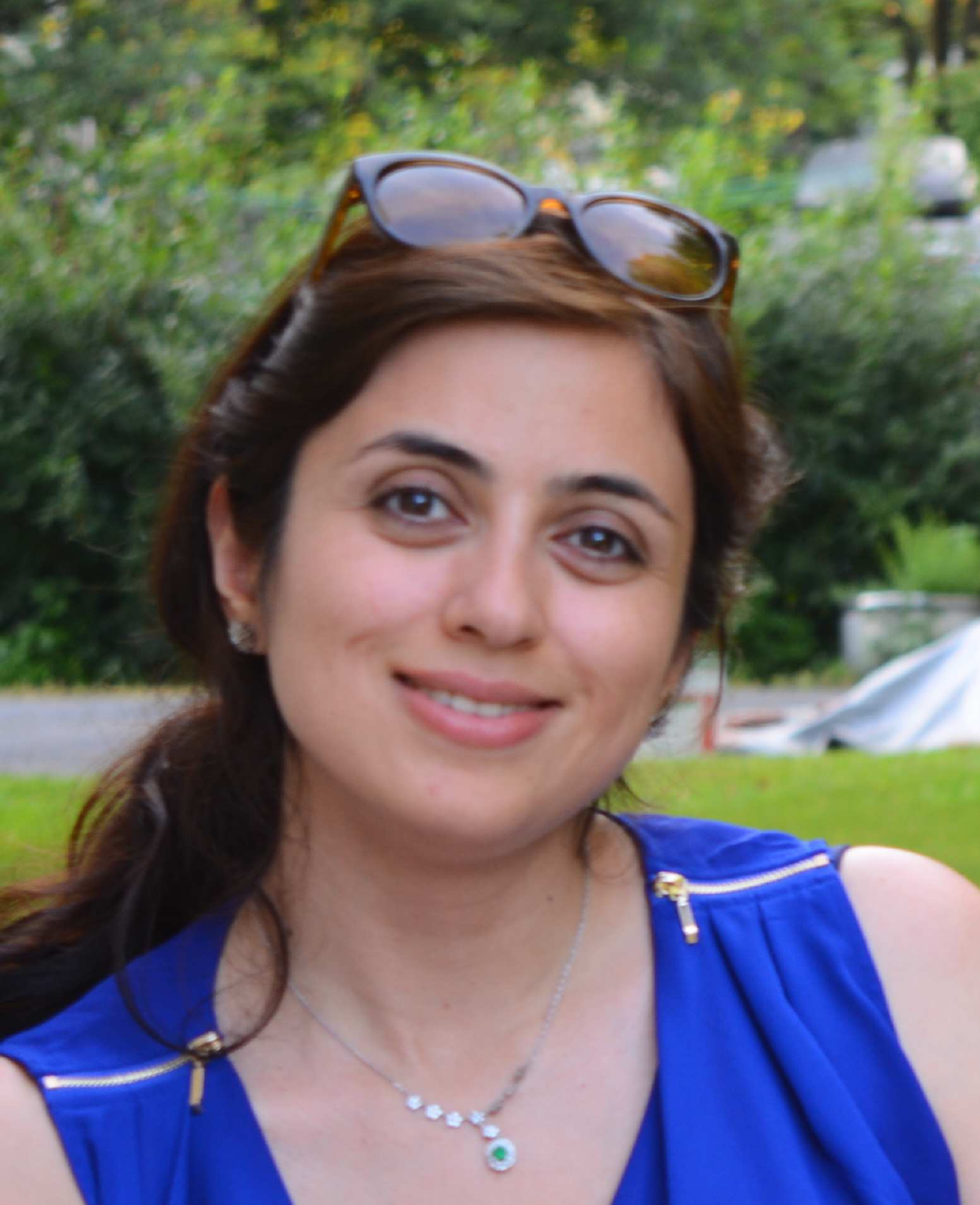}}]{Hana Khamfroush}
received her PhD with honors in 2014 in
telecommunications engineering from the University of Porto, Portugal. She then became 
a research associate at the computer science department of Penn State University. Since 2017 she is assistant professor at the computer science department of the College of Engineering at University of Kentucky.
She was named a rising star in EECS by MIT in 2015. Hana has served as TPC member and reviewer of several international 
conferences and Journals. She is currently the social media co-chair of IEEE N2Women community. Her research interests lie in the area of wireless networks and network management. 
\end{IEEEbiography}

\end{document}